\documentclass[12pt,draftclsnofoot,onecolumn]{IEEEtran}
\IEEEoverridecommandlockouts
\usepackage{graphicx}
\usepackage{amsmath}
\usepackage{amsfonts}
\usepackage{amssymb}
\usepackage{cite}
\usepackage{color}
\usepackage{multirow}
\usepackage{enumerate}
\usepackage{tabularx}
\usepackage{amsthm}
\usepackage{booktabs}
\usepackage[figurename=Fig., tablename=Tab.]{caption}

\usepackage{algorithmicx}
\usepackage{algorithm}
\usepackage{multicol}
\usepackage{algpseudocode}
\newtheorem{lem}{Lemma}

\newtheorem{theorem}{Theorem}
\newtheorem{define}{Definition}
\newtheorem{remark}{Remark}
\newcolumntype{L}[1]{>{\raggedright\arraybackslash}p{#1}}
\newcolumntype{C}[1]{>{\centering\arraybackslash}p{#1}}
\newcolumntype{R}[1]{>{\raggedleft\arraybackslash}p{#1}}

\usepackage{times, epsfig}
\usepackage{subfig}

\makeatletter
\newcommand{\multiline}[1]{%
	\begin{tabularx}{\dimexpr\linewidth-\ALG@thistlm}[t]{@{}X@{}}
		#1
	\end{tabularx}
}
\makeatother

\newlength{\figwidth}
\begin{document}
	
	\setlength{\pdfpagewidth}{8.5in}
	\setlength{\pdfpageheight}{11in}
\title{\LARGE Compute-Forward Multiple Access for Gaussian Fast Fading Channels}


\author{Lanwei Zhang,~\IEEEmembership{Student Member,~IEEE,}
        Jamie Evans,~\IEEEmembership{Senior Member,~IEEE,}
        and~Jingge Zhu,~\IEEEmembership{Member,~IEEE.}
}

\maketitle

\begin{abstract}
Compute-forward multiple access (CFMA) is a transmission strategy which allows the receiver in a multiple access channel (MAC) to first decode linear combinations of the transmitted signals and then solve for individual messages. Compared to existing MAC strategies such as joint decoding or successive interference cancellation (SIC), CFMA was shown to achieve the MAC capacity region for fixed channels under certain signal-to-noise (SNR) conditions without time-sharing using only single-user decoders. This paper studies the CFMA scheme for a two-user Gaussian fast fading MAC with channel state information only available at the receiver (CSIR). We develop appropriate lattice decoding schemes for the fading MAC and derive the achievable rate pairs for decoding linear combinations of codewords with any integer coefficients. We give a sufficient and necessary condition under which the proposed scheme can achieve the ergodic sum capacity. Furthermore, we investigate the impact of channel statistics on the capacity achievability of the CFMA scheme. In general, the sum capacity is achievable if the channel variance is small compared to the mean value of the channel strengths. Various numerical results are presented to illustrate the theoretical findings.
\end{abstract}


\section{Introduction}
Due to the broadcast nature of the medium, interference from other transmitters can be problematic for a receiver to recover the desired message in wireless communications. To solve this, traditional orthogonal transmission schemes allocate the transmission resources orthogonally to every transmitter to avoid interference. For example, in the time-division multiple access (TDMA) scheme each transmitter has its specialised time slot during which other transmitters will be silent. However, orthogonal schemes suffer from a diminishing rate, especially when many transmitters send messages simultaneously \cite{Tse05fundamentalsof}. 

Compute-and-forward (CF) is a non-orthogonal linear physical-layer network coding scheme proposed in \cite{Nazer11CF}, which allows the receiver to recover linear combinations of the messages from multiple transmitters. Unlike traditional schemes, which avoid interference or treat interference as noise, CF exploits interference in the decoding process by decoding linear combinations of codewords. Then the receiver can either forward the linear combination to the next receiver, thus working as a relay, or recover enough linear combinations to solve for individual messages. The CF scheme has been applied to different AWGN network models and results in good performance that cannot be achieved otherwise\cite{Nazer11CF,Zhu17,nazer16expand,Hong13,Ordentlich17,Zhan14,He18,Lyu19}. The key to the CF scheme is the adoption of nested lattice codes, which guarantee that any integer linear combination of the codewords is still a codeword. Due to the favourable algebraic property of lattice codes, they have been shown to approach or achieve the capacity in many cases\cite{Erez04,Ordentlich15,Hindy15,Campello16}. 

Compute-forward multiple access (CFMA) is a generalized CF scheme proposed in \cite{Zhu17}, which allows the users to have different rates whereas each user shares the same rate in the original CF. It shows that in a two-user Gaussian MAC with fixed channel gains, the CFMA scheme can achieve the entire capacity region without time sharing or rate splitting if the received SNR is large enough. It is not a straightforward task to extend the CFMA scheme to the fading case. CFMA and the original compute-and-forward scheme use lattice decoding (quantization with respect to lattice) which relies on the algebraic structure of lattice codes. The first observation in the fading case is that since each symbol in a codeword experiences a possibly different channel gain, the algebraic structure of lattice codes will be corrupted from the decoder's perspective, so it is not immediately clear how to extend the lattice decoder to the fading case. A second difference is that when applying CFMA to the fixed channel case, the lattice codes should be scaled appropriately according to the channel gains in order to achieve the entire capacity region. In the fading case with CSIR only, this is not possible as instantaneous channel realizations are not available at transmitters. As we will show later, it is still necessary to scale lattice codes in an appropriate way to achieve the capacity region, and this scaling can be made only depending on the statistics of the channel, which are assumed to be known at the transmitters. 

\subsection{Related work}

The original CF scheme is proposed in \cite{Nazer11CF} for AWGN networks with fixed channel gains. The achievable rates are proved to be optimal in the large signal-to-noise ratio (SNR) regime. The original CF scheme is expanded in \cite{nazer16expand,Zhu17}, which allows unequal rates among the transmitters. In both papers, the authors show that their proposed low-complexity scheme can achieve the channel capacity for the Gaussian multiple-access channel (MAC) given high SNR. Apart from the Gaussian MAC, CF is also applied to the random access Gaussian channel with a relatively large number of active users\cite{Ordentlich17}. It is shown that the energy-per-bit required by each user is significantly smaller than the popular solutions such as treating interference as noise. Moreover, the reverse-CF scheme has been proposed for the Gaussian broadcast channel as the dual of the MAC\cite{Hong13,He18}. In this scheme, a precoding process is required at the base station. Since the receiver will recover the integer linear combination of the transmitted signals, all the individual messages can be recovered through enough independent combinations. Following this idea, the integer forcing linear receiver has been proposed in \cite{Zhan14}. This linear receiver has low complexity and significantly outperforms conventional linear architectures such as zero forcing. It can also apply to the multiple-input multiple-output (MIMO) channel with no coding across transmit antennas while still achieving the optimal diversity-multiplexing tradeoff. 

The CF scheme is based on nested lattice codes. Lattice coding and decoding have been used and shown to provide good performance in many AWGN networks. It has been shown in\cite{Erez04} that they achieve the capacity of the point-to-point AWGN channel. The class of LAttice Space–Time (LAST) codes has been proposed in \cite{Gamal04} and is shown to achieve the optimal diversity–multiplexing tradeoff of multiple-input multiple-output (MIMO) channels. In \cite{Ordentlich15}, the authors show that with a precoded integer-forcing scheme, which is also based on nested lattice codes, they can achieve the optimal MIMO capacity to within a constant gap. Apart from these, there is also some literature discussing how lattice codes can be applied in the fading channel. In \cite{BechlawiRayleigh16}, Reed-Muller lattice coding is proposed for the Rayleigh block fading channel with an outstanding achievable rate. Lattice coding has also been applied for Rician fading channels in \cite{AlexRician17} and for fading dirty paper channels in \cite{AhmedDPC16}, where in both cases it gives better performance than classical coding schemes. More importantly, lattice codes can achieve ergodic capacity in fading channels. Algebraic lattices are applied to achieve the capacity of the ergodic fading channel in \cite{Campello16}. Similar results come from \cite{Hindy15} where the ergodic capacity is achieved with lattice codes. 

However, most of the literature considers the point-to-point fading channel case. When considering the fading MAC case, \cite{Hindy15} studies a two-user case with successive interference cancellation (SIC) decoding, which shows that with lattice codes the ergodic capacity is achievable. It is worth noting that SIC is a special case of CF. There are a few results regarding CF in fading channels. The outage performance of CF under slow fading is discussed in \cite{Nazer11CF} while only some numerical examples are provided there. Algebraic lattices are applied to block-fading channels in \cite{Lyu19} showing that they achieve better rates than the original CF scheme. Different to these prior works, the current paper focuses on the capacity achievability of the CFMA scheme. 

\subsection{Paper contributions}

In this paper, we investigate the CFMA scheme for the Gaussian fast fading multiple access channel with CSIR only. The main contributions of this paper are as follows.
\begin{itemize}
    \item CFMA is extended to fading channels with new technical contributions including the analysis of the ambiguity decoding of lattice codes in the CFMA scheme. The new result subsumes known results for CFMA for fixed channels and SIC decoding with lattice codes for the Gaussian MAC.
    \item We identify the channel conditions in which the capacity can be achieved. It is shown that under appropriate channel conditions, the entire ergodic capacity is achievable. A key to the result is a scaling parameter in the code construction that should be chosen judiciously based on the channel statistics.
    \item We investigate the impact of channel statistics. It is found that in general large mean and small variance of the channel gains benefit the capacity achievability. The findings are illustrated by the numerical examples.
\end{itemize}

\section{System model \& Coding scheme}\label{SM}

We consider a two-user Gaussian multiple access channel (MAC) with real-valued channel gains. We will consider a fast fading channel, where each codeword will see multiple channel realizations. For $n$ channel uses, the channel output at the receiver will be given by 
\begin{equation*}
    \textbf{y} = \sum_{l=1}^2 \textbf{H}_l \textbf{x}_l + \textbf{z},
\end{equation*}
where $\textbf{x}_l,\textbf{y}\in\mathbb{R}^n$ are the channel input of user $l$ and the channel output, respectively. The channel noise $\textbf{z}$ is assumed to be white Gaussian with zero mean and unit variance. The channel gains $\textbf{H}_l\in\mathbb{R}^{n\times n}$ are diagonal with diagonal entries representing the channel gains at each channel use. Note that $n$ is also the codeword length. The average input power constraint is assumed to be $P$ for each user. Note that if unequal transmission power is applied for each user, i.e., $P_l$ for user $l$, we can still keep this assumption by introducing the effective channel matrices $\textbf{H}_l' = P_l \textbf{H}_l/P$.

\subsection{Encoding}
The codebook of user $l$ is constructed by $n$-dimensional nested lattices denoted by
\begin{equation}\label{eq_codebook_MIMO}
    \mathcal{C}_l = \Lambda_l^F\cap \mathcal{V}_l^C,
\end{equation}
where the fine lattice $\Lambda_l^F$ is chosen to be \textit{good for AWGN channel coding} in the sense of \cite{Erez05}, and $\mathcal{V}_l^C$ is the Voronoi region of the coarse lattice $\Lambda_l^C$, which is chosen to be \textit{good for quantization}  in the sense of \cite{Erez05} with second moment
\begin{equation}\label{eq_second_moment_MIMO}
    \sigma^2(\Lambda_l^C) = \frac{1}{n\text{Vol}(\mathcal{V}_l^C)}\int_{\mathcal{V}_l^C} ||\textbf{x}||^2d\textbf{x} = \beta_l^2.
\end{equation}
With this codebook, the message rate of user $l$ is given by
\begin{equation}\label{eq_message_rate_MIMO}
    r_l = \frac{\log|\mathcal{C}_l|}{n} = \frac{1}{n}\log\frac{\text{Vol}(\mathcal{V}_l^C)}{\text{Vol}(\mathcal{V}_l^F)}.
\end{equation}
We introduce the dither vector $\textbf{d}_l$ that is uniformly distributed on the Voronoi region of $\Lambda_l^C / \beta_l$. The dithered unit-power codeword of $\textbf{t}_l\in \mathcal{C}_l$ is generated as
\begin{equation}
    \label{eq_c_l}
    \begin{split}
        \textbf{c}_l & = [\textbf{t}_l/\beta_l+\textbf{d}_l]\mod \Lambda_l^C/\beta_l \\
        & = (\textbf{t}_l/\beta_l+\textbf{d}_l) - \mathcal{Q}_{\Lambda_l^C /\beta_l}(\textbf{t}_l/\beta_l+\textbf{d}_l),
    \end{split}
\end{equation}
where $\mathcal{Q}_\Lambda$ is the quantization operation over the lattice $\Lambda$. Readers are referred to \cite{Nazer11CF,Zhu17} for more details about nested lattice codes. The channel input of user $l$ is then given by
\begin{equation}\label{eq_channel_input}
    \textbf{x}_l = \sqrt{P} \textbf{c}_l.
\end{equation}

\subsection{Decoding}
In the CFMA scheme, the receiver will decode enough integer linear combinations of the lattice codewords to recover the individual messages. In the two-user case, two linearly independent combinations are required and each of them is decoded with a properly chosen equalization matrix. When multiplied by the equalization matrix, the channel output can be rewritten as a summation of the desired linear combination and the effective noise. Since the decoded linear combinations are still lattice codewords, we can use a single-user lattice decoder for each linear combination. In this paper, we apply the ambiguity lattice decoder of \cite{Loeliger97}, which is defined in Definition~\ref{def_ambuiguity_lattice_decoder}. The detailed decoding process can be found in the proof of Theorem~\ref{thm_achievable_rate}.
\begin{define}[Ambiguity lattice decoder]\label{def_ambuiguity_lattice_decoder}
    For a $n$-dimensional lattice $\Lambda\subset \mathbb{R}^n$, the ambiguity lattice decoder with a decision region $\mathcal{E}\in\mathbb{R}^n$ gives the result $\textbf{u}\in\mathbb{R}^n$ if the obtained signal  $\textbf{r}\in\Lambda$ can be written in the form $\textbf{r} = \textbf{u} + \textbf{z}$ with a unique lattice codeword $\textbf{u}\in\Lambda$ and the noise $\textbf{z}\in \mathcal{E}$. Otherwise it produces an ambiguity error.
\end{define}

\section{Achievable rates}
In this section, we will first look at a specifically characterized noisy output $\textbf{r}\in\mathbb{R}^n$ in Lemma~\ref{lem_rate_linear_combination} and show the existence of a lattice codebook construction with certain message rates such that the the error probability of decoding a lattice codeword $\textbf{u}\in\Lambda_F$ from $\textbf{r}$ could be arbitrarily small. Here $\Lambda_F$ stands for the finer lattice of $\Lambda_1^F$ and $\Lambda_2^F$. With this auxiliary result, we will derive the expression of the achievable rate pair to decode individual messages. Note that $|M|$ is used to represent the determinant of matrix $M$.

\begin{lem}\label{lem_rate_linear_combination}
    Let $\textbf{r} = \textbf{u} + \textbf{A}\textbf{z} + \sum_{l=1}^2 \textbf{B}_l \textbf{c}_l$ where $\textbf{u}\in\Lambda_F$ is the lattice codeword, $\textbf{z}\in\mathcal{N}(\textbf{0},\textbf{I}_n)$ is the channel noise, $\textbf{c}_l$ is defined in (\ref{eq_c_l}) and $\textbf{A}, \textbf{B}_l$ are non-zero square matrices. There exists a lattice codebook construction in the form of (\ref{eq_codebook_MIMO}) with the rate defined in (\ref{eq_message_rate_MIMO}) satisfying
    \begin{equation}\label{eq_achievable_rate_linear_combination}
        r_l < \lim_{n\rightarrow \infty}\; \max_{\textbf{A},\textbf{B}_l}\; \frac{1}{2n} \log^+ \frac{\beta_l^{2n}}{\left|\textbf{A}\textbf{A}^T + \sum_{l=1}^2 \textbf{B}_l\textbf{B}_l^T\right|},
    \end{equation}
    such that with the ambiguity lattice decoder in Definition~\ref{def_ambuiguity_lattice_decoder}, the decoding error could be arbitrarily small. 
\end{lem}

\begin{proof}
    For brevity, we let $\boldsymbol{\Sigma} = \textbf{A}\textbf{A}^T + \sum_{l=1}^2 \textbf{B}_l\textbf{B}_l^T$. In the proof, $Pr(E)$ is used to represent the probability of the event $E$. Recall that we are considering an ensemble of $n$-dimensional nested lattices $\{\Lambda_l^C \subseteq \Lambda_l^F\}$ with the message rate given by (\ref{eq_message_rate_MIMO}). We let $\Lambda_C = \Lambda_l^C / \beta_l$ thus $\sigma^2(\Lambda_C) = 1$ according to (\ref{eq_second_moment_MIMO}). Note that $\textbf{c}_l$ is uniformly distributed on $\mathcal{V}_l^C/\beta_l$ by definition. We further let $\Lambda_F = \Lambda_l^F$ for $l=1,2$ and denote its Voronoi region as $\mathcal{V}_F$. For this lattice $\Lambda_F$, we apply the ambiguity lattice decoder with the decision region
    \begin{equation}
        \label{eq_decision_region}
        \mathcal{E}_{n,\eta} = \{\textbf{z}\in\mathbb{R}^{n}: |\textbf{Q}\textbf{z}|^2\leq n(1+\eta)\}
    \end{equation}
    for some $\eta>0$ where $\textbf{Q}\in\mathbb{R}^{n\times n}$ satisfies $\textbf{Q}^T\textbf{Q} = \boldsymbol{\Sigma}^{-1}$. The decomposition of $\boldsymbol{\Sigma}^{-1}$ is always possible since $\boldsymbol{\Sigma}$ is positive-definite. We rewritten $\textbf{r}$ as $\textbf{r} = \bar{\textbf{z}} + \textbf{u}$ with the effective noise $\bar{\textbf{z}} = \textbf{A}\textbf{z} + \sum_{l=1}^L \textbf{B}_l \textbf{c}_l$. By definition, $\textbf{u}$, $\textbf{z}$ and $\textbf{c}_l$ are independent of each other. From Definition~\ref{def_ambuiguity_lattice_decoder}, the decoding error event $E_\mathcal{D}$ happens when the effective noise $\bar{\textbf{z}}$ is outside the decision region $\mathcal{E}_{n,\eta}$ or $\bar{\textbf{z}}\in\mathcal{E}_{n,\eta}$ but an ambiguity event occurs. The ambiguity event $E_\mathcal{A}$ is defined by $\textbf{r}$ belonging to $(\mathcal{E}_{n,\eta} + \textbf{u}_1) \cap (\mathcal{E}_{n,\eta} + \textbf{u}_2)$ for some pair of distinct lattice points $\textbf{u}_1,\textbf{u}_2\in\Lambda_F$. Note that the probability of an ambiguity does not depend on the codeword $\textbf{u}$.
    By the union bound, the error probability of decoding $\textbf{u}$ from $\textbf{r}$ for a given $\Lambda_F$ is upper-bounded by
    \begin{equation}
        \label{eq_error_prob}
        P_e: = Pr(E_\mathcal{D}) \leq Pr(\bar{\textbf{z}} \notin \mathcal{E}_{n,\eta}) + Pr(E_\mathcal{A}|\bar{\textbf{z}} \in \mathcal{E}_{n,\eta}).
    \end{equation}
    We will then take the expectation over the ensemble of random lattices. Since $Pr(\bar{\textbf{z}} \notin \mathcal{E}_{n,\eta})$ is independent of $\Lambda_F$, we can use the ``random coding" theorem of \cite[Theorem 4]{Loeliger97} to upper bound the average decoding error probability as
    \begin{equation}
        \label{eq_average_error_prob}
        \bar P_e \leq Pr(\bar{\textbf{z}} \notin \mathcal{E}_{n,\eta}) + (1+\delta)\frac{\text{Vol}(\mathcal{E}_{n,\eta})}{\text{Vol}(\mathcal{V}_F)},
    \end{equation}
    for any $\delta>0$. Let $Pr(\bar{\mathcal{A}}) = (1+\delta)\frac{\text{Vol}(\mathcal{E}_{n,\eta})}{\text{Vol}(\mathcal{V}_F)}$. The term $\text{Vol}(\mathcal{E}_{n,\eta})$ refers to the volume of the decision region which can be represented by
    \begin{equation}
        \label{eq_vol_decision_region}
        \text{Vol}(\mathcal{E}_{n,\eta}) = (1+\eta)^{n/2}|\textbf{Q}^T\textbf{Q}|^{-1/2}\text{Vol}(\mathcal{B}(\sqrt{n})),
    \end{equation}
    where $\text{Vol}(\mathcal{B}(\sqrt{n}))$ is the volume of a $n$-dimensional sphere of radius $\sqrt{n}$. From (\ref{eq_message_rate_MIMO}), we know 
    \begin{equation}
        \label{eq_vol_F}
        \text{Vol}(\mathcal{V}_F) = \text{Vol}(\mathcal{V}_l^C)\cdot2^{-n r_l} = \beta_l^{n}\text{Vol}(\mathcal{V}_C)\cdot2^{-n r_l}.
    \end{equation}
    By combining (\ref{eq_vol_decision_region}) and (\ref{eq_vol_F}), we have
    \begin{equation}
        \label{eq_ambiguity_prob_rate}
        Pr(\bar{\mathcal{A}}) \leq (1+\delta)\cdot 2^{-n\left[\frac{1}{2n}\log\frac{\beta_l^{2n}}{|\boldsymbol{\Sigma}|}-r_l-\eta'\right]}
    \end{equation}
    where 
    \begin{equation}
        \label{eq_eta_prime}
        \eta' = \frac{1}{2}\log(1+\eta)+\frac{1}{n}\log \frac{\text{Vol}(\mathcal{B}(\sqrt{n}))}{\text{Vol}(\mathcal{V}_C)}.
    \end{equation}
    As $n\rightarrow\infty$, we have $G(\Lambda_C)\rightarrow \frac{1}{2\pi e}$ due to the fact that the constructed coarse lattices are good for covering, where $G(\Lambda_C)$ is the normalized second-order moment of $\Lambda_C$. Given $\text{Vol}(\mathcal{V}_C) = \left(\frac{\sigma^2(\Lambda_C)}{G(\Lambda_C)}\right)^{n/2}$, we have $\text{Vol}(\mathcal{V}_C) \rightarrow (2\pi e)^{n/2}$ when $n\rightarrow\infty$. By using the fact
    \begin{equation}
        \label{eq_vol_ball_inf}
        \text{Vol}(\mathcal{B}(\sqrt{n})) \rightarrow \frac{(2\pi e)^ {n/2}}{\sqrt{n\pi}} \text{ as } n\rightarrow\infty, 
    \end{equation}
    we have when $n\rightarrow\infty$,
    \begin{equation}
        \label{eq_vol_fraction_inf}
        \frac{1}{n}\log \frac{\text{Vol}(\mathcal{B}(\sqrt{n}))}{\text{Vol}(\mathcal{V}_C)} \rightarrow -\log (\sqrt{n\pi})^{1/n} \rightarrow 0.
    \end{equation}
    The second right arrow of (\ref{eq_vol_fraction_inf}) comes from the fact that $\lim_{x\rightarrow\infty}(x)^{\lambda/x}=1,\forall\lambda>0$ when $x = n\pi$ and $\lambda = \frac{\pi}{2}$. Based on (\ref{eq_vol_fraction_inf}) and the fact that $\eta>0$ is arbitrary, $\eta'$ can be made arbitrarily small for large enough $n$. Therefore, we can conclude that for arbitrary $\epsilon_1 >0$, $Pr(\bar{\mathcal{A}}) \leq \epsilon_1/2$ for sufficiently large $n$ given that 
    \begin{equation}
        \label{eq_rate_ambiguity}
        r_l < \frac{1}{2n}\log\frac{\beta_l^{2n}}{|\boldsymbol{\Sigma}|}.
    \end{equation}
    The proof will then be complete if we can show that $Pr(\bar{\textbf{z}} \notin \mathcal{E}_{n,\eta})\leq \epsilon_2/2$ for arbitrary $\eta, \epsilon_2>0$. Recall that 
    \begin{equation}
        \label{eq_wrong_decode_error}
        Pr(\bar{\textbf{z}} \notin \mathcal{E}_{n,\eta}) = Pr(|\textbf{Q}\bar{\textbf{z}}|^2 > n(1+\eta)),
    \end{equation}
    where $\bar{\textbf{z}} = \textbf{A}\textbf{z} + \sum_{l=1}^2 \textbf{B}_l \textbf{c}_l$ is the effective noise with $\textbf{z}\sim \mathcal{N}(\textbf{0},\textbf{I}_{n})$ and $\textbf{c}_l\sim \text{Uniform}(\mathcal{V}_C)$ for $l=1,2$. Note that $\textbf{z}$ and $\textbf{c}_l$'s are statistically independent. To upper-bound this probability, we will consider a ``noisier" system with higher noise variance. We first add a Gaussian vector $\textbf{e}_3\sim \mathcal{N}(\textbf{0},(\sigma^2-1)\textbf{I}_{n})$ to $\textbf{z}$ to make it noisier, where 
    \begin{equation}
        \sigma^2 = \frac{r_{c}(\Lambda_C)^2}{n}.
    \end{equation}
    The term $r_{c}(\Lambda_C)$ is the covering radius of $\Lambda_C$. From \cite{Gamal04}, we know $\sigma^2>1$. As $\sigma^2-1>0$, the noise $\textbf{e}_3$ is well-defined. We then replace $\textbf{c}_l$ with Gaussian vector $\textbf{e}_l\sim\mathcal{N}(\textbf{0},\sigma^2\textbf{I}_{n})$. Recall $\sigma^2(\Lambda_C) =1$, thus $\textbf{e}_l$ has a larger variance. The considered noise is then replaced by
    \begin{equation}
        \label{eq_noisier_noise}
        \bar{\textbf{z}}' = \textbf{A}(\textbf{z}+\textbf{e}_3) + \sum_{l=1}^2 \textbf{B}_l\textbf{e}_l.
    \end{equation}
    Let $f_{\textbf{c}_l}(\cdot)$ and $f_{\textbf{e}_l}(\cdot)$ denote the PDF of $\textbf{c}_l$ and $\textbf{e}_l$, respectively. Following the argument of \cite[Lemma 11]{Erez04} and \cite{Gamal04}, we obtain
    \begin{equation}
        \label{eq_pdf_upper_bound}
        f_{\textbf{c}_l}(\textbf{z}) \leq \left(\frac{r_{c}(\Lambda_C)}{r_{e}(\Lambda_C)}\right)^{n} \exp(o(n)) f_{\textbf{e}_l}(\textbf{z}),
    \end{equation}
    where $r_{e}(\Lambda_C)$ is the effective radius of $\Lambda_C$. When $n\rightarrow\infty$, $r_{c}(\Lambda_C)/r_{e}(\Lambda_C)\rightarrow 1$ from the goodness for covering. From the noisier construction, the RHS of (\ref{eq_wrong_decode_error}) can be upper-bounded by
    \begin{equation}
        \label{eq_wrong_decode_error_upper_bound}
        \begin{split}
            Pr\left(|\textbf{Q}\bar{\textbf{z}}|^2 > n(1+\eta)\right) \leq \left[\left(\frac{r_{c}(\Lambda_C)}{r_{e}(\Lambda_C)}\right)^{n} \exp(o(n))\right]^L  Pr(|\textbf{Q}\bar{\textbf{z}}'|^2 \geq n(1+\eta)).
        \end{split}
    \end{equation}
    Note that $\textbf{Q}\bar{\textbf{z}}'\sim \mathcal{N}(\textbf{0},\sigma^2\textbf{I}_{n})$. Thus, $|\textbf{Q}\bar{\textbf{z}}'/\sigma|^2\sim \chi^2(n)$. Similar to the approach taken in \cite{Gamal04}, we use the Chernoff bound to upper-bound $Pr(|\textbf{Q}\bar{\textbf{z}}'|^2 > n(1+\eta))$, which gives
    \begin{equation}
        \label{eq_wrong_decode_error_chernoff_bound}
        Pr(|\textbf{Q}\bar{\textbf{z}}'|^2 \geq n(1+\eta)) \leq \min_{\alpha > 0}\; e^{-\alpha n(1+\eta)} \mathbb{E}[e^{\alpha|\textbf{Q}\bar{\textbf{z}}'|^2}],
    \end{equation}
    where $\mathbb{E}[e^{\alpha|\textbf{Q}\bar{\textbf{z}}'|^2}] = \exp(-\frac{n}{2}\ln(1-2\alpha \sigma^2))$. Therefore,
    \begin{equation}
        \label{eq_wrong_decode_error_chernoff_bound_min}
        \begin{split}
            Pr(|\textbf{Q}\bar{\textbf{z}}'|^2 \geq n(1+\eta)) & \leq \min_{\alpha > 0}\; e^{-\frac{n}{2}[2\alpha (1+\eta)+\ln(1-2\alpha \sigma^2))]}\\
            & = \exp\left(-\frac{n}{2}(\zeta - \ln \zeta - 1)\right),
        \end{split}
    \end{equation}
    where $\zeta = \frac{1+\eta}{\sigma^2}$. It is worth noting that $\sigma^2\rightarrow 1$ when $n\rightarrow\infty$. For arbitrary $\eta >0$, we have $\zeta > 1$ which leads to $\zeta - \ln \zeta - 1 >0$. We can then conclude that $Pr(\bar{\textbf{z}} \notin \mathcal{E}_{n,\eta})\leq \epsilon_2/2$ for arbitrary $\eta, \epsilon_2>0$ and sufficiently large $n$. Then the proof is complete.
\end{proof}

Now we present the achievable rate pair for decoding individual messages in Theorem~\ref{thm_achievable_rate}, where we decode two linear combinations with coefficients $\textbf{a}$ and $\textbf{b}$, respectively.
\begin{theorem}\label{thm_achievable_rate}
For a two-user Gaussian fast fading MAC, the CFMA scheme gives the following achievable rate pair for decoding individual messages
\begin{equation}\label{eq_achievable_rate_ergodic_MAC_CFMA}
    R_l = \left\{\begin{array}{ll}
        r_l(\textbf{a},\boldsymbol{\beta}), & b_l = 0 \\
        r_l(\textbf{b}|\textbf{a},\boldsymbol{\beta}), & a_l = 0 \\
        \min\{r_l(\textbf{a},\boldsymbol{\beta}),\;r_l(\textbf{b}|\textbf{a},\boldsymbol{\beta})\}, & \text{otherwise}
    \end{array}\right.
\end{equation}
for any linearly independent $\textbf{a},\textbf{b}\in\mathbb{Z}^2$ and $\boldsymbol{\beta}\in\mathbb{R}^2$ if $r_l(\textbf{a},\boldsymbol{\beta})\geq 0$ and $r_l(\textbf{b}|\textbf{a},\boldsymbol{\beta}) \geq 0$ for $l = 1,2$, where
\begin{align}
    r_l(\textbf{a},\boldsymbol{\beta}) = & \frac{1}{2} \mathbb{E}_{\textbf{h}} \left[\log 
    \frac{\beta_l^{2} \left(1 + P h_{1}^2 + P h_{2}^2\right)}
    {M(\boldsymbol{\beta},\textbf{h})}\right] \label{eq_achievable_rate_MIMO_first}\\
    r_l(\textbf{b}|\textbf{a},\boldsymbol{\beta}) = & \frac{1}{2} \mathbb{E}_{\textbf{h}} \left[ \log \frac{\beta_l^{2} { M(\boldsymbol{\beta},\textbf{h}) }}
    {(\tilde a_1 \tilde b_2 - \tilde a_2 \tilde b_1)^{2}} \right]\label{eq_achievable_rate_MIMO_second}
\end{align}
with $\tilde a_l = a_l \beta_l, \tilde b_l = b_l \beta_l, \textbf{h} = (h_{1},h_{2})^T$ and 
\begin{equation}
    \label{eq_Mh}
    M(\boldsymbol{\beta},\textbf{h}) = (\tilde a_1^2 + \tilde a_2^2) + P (\tilde a_1 h_{2} - \tilde a_2 h_{1})^2.
\end{equation}
\end{theorem}
\begin{proof}
To decode both individual messages, two linearly independent combinations are required, whose coefficients are denoted by $\textbf{a}$ and $\textbf{b}$, respectively. Let $\textbf{h}_i = (h_{1i},h_{2i})$ denote the channel gain for the $i$-th channel use. When decoding the first linear combination with coefficients $\textbf{a}$, the receiver computes
\begin{align*}
    \textbf{y}_1^\prime = & \textbf{W} \textbf{y} -\sum_{l=1}^2 a_l\beta_l \textbf{d}_l\\
    = & \textbf{W}(\sum_{l=1}^2 \sqrt{P} {\textbf{H}}_l \textbf{c}_l + \textbf{z})-\sum_{l=1}^2 a_l\beta_l \textbf{d}_l -\sum_{l=1}^2 a_l\beta_l \textbf{c}_l+\underbrace{\sum_{l=1}^2 a_l\beta_l \textbf{c}_l}_{\text{use } (\ref{eq_c_l})}\\
    = & \underbrace{\textbf{W}\textbf{z} + \sum_{l=1}^2 (\tilde a_l\textbf{I}_{n}-\sqrt{P}\textbf{W}{\textbf{H}}_l)(-\textbf{c}_l)}_{\bar{\textbf{z}}_1}-\sum_{l=1}^2 a_l\beta_l \textbf{d}_l\\
    & + \sum_{l=1}^2 a_l\beta_l(\textbf{t}_l/\beta_l+\textbf{d}_l) - a_l\beta_l\mathcal{Q}_{\Lambda_l^C /\beta_l}(\textbf{t}_l/\beta_l+\textbf{d}_l) \\
    = & \bar{\textbf{z}}_1 + \sum_{l=1}^2 a_l \underbrace{\left(\textbf{t}_l - \mathcal{Q}_{\Lambda_l^C }(\textbf{t}_l+\beta_l\textbf{d}_l)\right)}_{\tilde{\textbf{t}}_l},
\end{align*}
where $\textbf{W}$ is a $n\times n$ equalization matrix which will be determined later. The last equality holds because $\mathcal{Q}_\Lambda(\beta \textbf{x}) = \beta \mathcal{Q}_{\Lambda/\beta}(\textbf{x})$ for any lattice $\Lambda$ and any nonzero real $\beta$. Note that $\tilde{\textbf{t}}_l\in\Lambda_l^F$ is not the codeword of user $l$. But it is enough to recover the codewords since $\tilde{\textbf{t}}_l$ and the codeword $\textbf{t}_l$ belong to the same coset of $\Lambda_l^C$ \cite{Zhu17,Gamal04}. Since $\textbf{c}_l$ is uniformly distributed on $\mathcal{V}_l^C/\beta_l$ which is symmetric with respect to the origin, the term $-\textbf{c}_l$ in $\bar{\textbf{z}}_1$ has the same distribution as $\textbf{c}_l$. Following Lemma~\ref{lem_rate_linear_combination} with $\textbf{u} = \sum_{l=1}^2 a_l\tilde{\textbf{t}}_l$, $\textbf{A} = \textbf{W}$ and $\textbf{B}_l = \tilde a_l\textbf{I}_{n}-\sqrt{P}\textbf{W}{\textbf{H}}_l$, the achievable rate of decoding the linear combination $\sum_{l=1}^2 a_l\tilde{\textbf{t}}_l$ is given by
\begin{equation}\label{eq_achievable_rate_general}
    r_{l}^{(1)} < \lim_{n\rightarrow\infty}\; \max_{\textbf{W}}\; \frac{1}{2n} \log^+ \frac{\beta_l^{2n}}{|\boldsymbol{\Sigma}_1(\textbf{W})|},
\end{equation}
where $\boldsymbol{\Sigma}_1(\textbf{W}) = \textbf{W}\textbf{W}^T + \sum_{l=1}^2 (\tilde a_l\textbf{I}_{n}-\sqrt{P}\textbf{W}{\textbf{H}}_l)(\tilde a_l\textbf{I}_{n}-\sqrt{P}\textbf{W}{\textbf{H}}_l)^T$. Note that $\textbf{H}_l$ is diagonal. Expanding $\boldsymbol{\Sigma}_1(\textbf{W})$ gives
\begin{equation}
    \label{eq_Sigma_1}
    \begin{split}
        \boldsymbol{\Sigma}_1(\textbf{W}) = & \sum_{l=1}^2\tilde a_l^2\textbf{I}_{n} +  \textbf{W}\left(\textbf{I}_{n} + P \sum_{l=1}^2 {\textbf{H}}_l^2\right)\textbf{W}^T - \left(\sum_{l=1}^2\tilde a_l \sqrt{P} {\textbf{H}}_l^T\right) \textbf{W}^T - \textbf{W} \left(\sum_{l=1}^2\tilde a_l \sqrt{P} {\textbf{H}}_l \right).
    \end{split}
\end{equation}
By ``completing the square'', we can derive the optimal $\textbf{W}$ as
\begin{equation}
    \label{eq_optimal_W}
    \textbf{W}^* = \left(\sum_{l=1}^2\tilde a_l \sqrt{P}{\textbf{H}}_l\right) \left(\textbf{I}_{n} + \sum_{l=1}^2 P {\textbf{H}}_l^2 \right)^{-1}.
\end{equation}
Thus, 
\begin{equation}
    \label{eq_optimal_Sigma_1}
    \begin{split}
        \boldsymbol{\Sigma}_1(\textbf{W}^*) = &\sum_{l=1}^2\tilde a_l^2\textbf{I}_{n} - \left(\sum_{l=1}^2\tilde a_l \sqrt{P}{\textbf{H}}_l\right)\left(\textbf{I}_{n} + \sum_{l=1}^2 P {\textbf{H}}_l^2 \right)^{-1}\left(\sum_{l=1}^2\tilde a_l \sqrt{P}{\textbf{H}}_l\right).
    \end{split}
\end{equation}
Using the matrix determinant lemma 
\begin{equation}\label{eq_matrix_determinant_lemma}
    |\textbf{D}-\textbf{C}\textbf{A}^{-1}\textbf{B}| = \frac{\left|\begin{array}{cc}
        \textbf{A} & \textbf{B} \\
        \textbf{C} & \textbf{D}
    \end{array}\right|}{|\textbf{A}|} = \frac{|\textbf{D}||\textbf{A}-\textbf{B}\textbf{D}^{-1}\textbf{C}|}{|\textbf{A}|},
\end{equation}
We can write $|\boldsymbol{\Sigma}_1(\textbf{W}^*)|$ as
\begin{equation}
    \label{eq_optimal_det_Sigma_1}
    |\boldsymbol{\Sigma}_1(\textbf{W}^*)| = \frac{\left|(\tilde a_1^2 + \tilde a_2^2)\textbf{I}_{n} + P (\tilde a_1 {\textbf{H}}_2 - \tilde a_2 \textbf{H}_1)^2 \right|}{\left|\textbf{I}_{n} + \sum_{l=1}^2 P \textbf{H}_l^2 \right|}.
\end{equation}
Since $\textbf{H}_1$ and $\textbf{H}_2$ are diagonal, it can be written as
\begin{equation}
    \label{eq_optimal_det_Sigma_1_prod}
    |\boldsymbol{\Sigma}_1(\textbf{W}^*)| = \prod_{i=1}^n \frac{(\tilde a_1^2 + \tilde a_2^2) + P (\tilde a_1 h_{2i} - \tilde a_2 h_{1i})^2}{1 + P h_{1i}^2 + P h_{2i}^2}.
\end{equation}
Plugging the above into (\ref{eq_achievable_rate_general}) gives
\begin{equation}\label{eq_rate1_lim}
    r_l^{(1)} < \lim_{n\rightarrow\infty}\; \frac{1}{2n} \sum_{i=1}^n \log 
    \frac{\beta_l^{2} \left(1 + P h_{1i}^2 + P h_{2i}^2\right)}
    {M(\boldsymbol{\beta},\textbf{h}_i)} = \frac{1}{2} \mathbb{E}_{\textbf{h}} \left[\log 
    \frac{\beta_l^{2} \left(1 + P h_{1}^2 + P h_{2}^2\right)}
    {M(\boldsymbol{\beta},\textbf{h})}\right],
\end{equation}
where $M(\boldsymbol{\beta},\textbf{h})$ is given by (\ref{eq_Mh}) and 
\begin{equation}
    \label{eq_Mhi}
    M(\boldsymbol{\beta},\textbf{h}_i) = (\tilde a_1^2 + \tilde a_2^2) + P (\tilde a_1 h_{2i} - \tilde a_2 h_{1i})^2.
\end{equation}

Given the first decoded linear combination, we can reconstruct $\sum_{l=1}^2 \tilde a_l \textbf{c}_l = \sum_{l=1}^2 a_l \tilde{\textbf{t}}_l + \sum_{l=1}^2 a_l \beta_l \textbf{d}_l$. To decode the second linear combination with coefficients $\textbf{b}$, the receiver computes
\begin{align*}
    \textbf{y}_2^\prime = & \textbf{F}\textbf{y} + \textbf{L} \sum_{l=1}^2 a_l\beta_l \textbf{c}_l - \sum_{l=1}^2 b_l \beta_l \textbf{d}_l \\
    = &  \textbf{F} \sum_{l=1}^2 \sqrt{P}{\textbf{H}}_l \textbf{c}_l + \textbf{F}\textbf{z} + \textbf{L} \sum_{l=1}^2 a_l\beta_l \textbf{c}_l - \sum_{l=1}^2 b_l\beta_l \textbf{d}_l - \sum_{l=1}^2 b_l\beta_l \textbf{c}_l + \sum_{l=1}^2 b_l\beta_l \textbf{c}_l \\
    = & \underbrace{\textbf{F} {\textbf{z}} + \sum_{l=1}^2 (\tilde b_l \textbf{I}_{n} - \sqrt{P}\textbf{F}{\textbf{H}}_l - \tilde a_l \textbf{L})(-\textbf{c}_l) }_{\bar{\textbf{z}}_2} + \sum_{l=1}^2 b_l \tilde{\textbf{t}},
\end{align*}
where $\textbf{F},\textbf{L}$ are $n\times n$ matrices that will be determined later. Following Lemma~\ref{lem_rate_linear_combination} with $\textbf{u} = \sum_{l=1}^2 b_l\tilde{\textbf{t}}_l$, $\textbf{A} = \textbf{F}$ and $\textbf{B}_l = \tilde b_l\textbf{I}_{n}-\sqrt{P}\textbf{F}{\textbf{H}}_l - \tilde a_l \textbf{L}$, the achievable rate of user $l$ is given by
\begin{equation}\label{eq_achievable_rate_b}
    r_{l}^{(2)} <\lim_{n\rightarrow\infty} \; \max_{\textbf{F},\textbf{L}}\; \frac{1}{2n} \log^+ \frac{\beta_l^{2n}}{|\boldsymbol{\Sigma}_2(\textbf{F},\textbf{L})|},
\end{equation}
where $\boldsymbol{\Sigma}_2(\textbf{F},\textbf{L}) = \textbf{F}\textbf{F}^T + \sum_{l=1}^2 (\tilde b_l\textbf{I}_{n}-\sqrt{P}\textbf{F}{\textbf{H}}_l - \tilde a_l \textbf{L})(\tilde b_l\textbf{I}_{n}-\sqrt{P}\textbf{F}{\textbf{H}}_l - \tilde a_l \textbf{L})^T$. Expanding $\boldsymbol{\Sigma}_2(\textbf{F},\textbf{L})$ gives
\begin{equation}
    \label{eq_Sigma_2}
    \begin{split}
    \boldsymbol{\Sigma}_2(\textbf{F},\textbf{L}) = & \sum_{l=1}^2\tilde b_l^2 \textbf{I}_{n} +  \textbf{F}\left(\textbf{I}_{n} + \sum_{l=1}^2 P {\textbf{H}}_l^2\right)\textbf{F}^T - \left(\sum_{l=1}^2\tilde b_l \sqrt{P}{\textbf{H}}_l\right) \textbf{F}^T - \textbf{F} \left(\sum_{l=1}^2\tilde b_l \sqrt{P} {\textbf{H}_l} \right) \\
    & + \sum_{l=1}^2 \tilde a_l^2\textbf{L}\textbf{L}^T - \left(\sum_{l=1}^2 \tilde a_l (\tilde b_l \textbf{I}_{n} - \sqrt{P} \textbf{F} {\textbf{H}}_l )\right) \textbf{L}^T - \textbf{L} \left(\sum_{l=1}^2 \tilde a_l  (\tilde b_l \textbf{I}_{n} - \sqrt{P}  {\textbf{H}}_l \textbf{F}^T)\right).
    \end{split}
\end{equation}
Optimizing over $\textbf{L}$ gives
\begin{equation}
    \label{eq_optimal_L}
    \textbf{L}^*(\textbf{F}) = \left(\sum_{l=1}^2\tilde a_l^2\right)^{-1}\left(\sum_{l=1}^2 \tilde a_l (\tilde b_l \textbf{I}_{n} - \sqrt{P} \textbf{F}{\textbf{H}}_l )\right).
\end{equation}
Plugging $\textbf{L}^*(\textbf{F})$ into (\ref{eq_Sigma_2}) and optimizing over $\textbf{F}$ leads to
\begin{equation}
    \label{eq_optimal_F}
    \textbf{F}^* = \frac{\tilde a_1 \tilde b_2 - \tilde a_2 \tilde b_1}{\sum_{l=1}^2\tilde a_l^2}
    \sqrt{P}(\tilde a_1 {\textbf{H}}_2 - \tilde a_2 {\textbf{H}}_1)
    \textbf{K}^{-1},
\end{equation}
where 
\begin{equation}
    \label{eq_Mk}
    \begin{split}
        \textbf{K} = & \textbf{I}_{n} + \frac{P}{\sum_{l=1}^2\tilde a_l^2} (\tilde a_1 {\textbf{H}}_2  - \tilde a_2 {\textbf{H}}_1)^2 .
    \end{split}
\end{equation}
We can then derive 
\begin{equation}
    \label{eq_optimal_Sigma_2}
    \begin{split}
        \boldsymbol{\Sigma}_2(\textbf{F}^*,\textbf{L}^*(\textbf{F}^*)) = \frac{(\tilde a_1 \tilde b_2 - \tilde a_2 \tilde b_1)^2}{\sum_{l=1}^2\tilde a_l^2} \left[\textbf{I}_{n} - \frac{P}{\sum_{l=1}^2\tilde a_l^2}  (\tilde a_1 {\textbf{H}}_2 - \tilde a_2 {\textbf{H}}_1) \textbf{K}^{-1} (\tilde a_1 {\textbf{H}}_2 - \tilde a_2 {\textbf{H}}_1)\right].
    \end{split}
\end{equation}
Applying Woodbury matrix identity gives
\begin{equation}
    \label{eq_optimal_Sigma_2_woodbury}
    \begin{split}
        \boldsymbol{\Sigma}_2(\textbf{F}^*,\textbf{L}^*(\textbf{F}^*)) = \frac{(\tilde a_1 \tilde b_2 - \tilde a_2 \tilde b_1)^2}{\sum_{l=1}^2\tilde a_l^2} \textbf{K}^{-1}.
    \end{split}
\end{equation}
Therefore,
\begin{align}
    |\boldsymbol{\Sigma}_2(\textbf{F}^*,\textbf{L}^*(\textbf{F}^*))|
    & = \left| \frac{(\tilde a_1 \tilde b_2 - \tilde a_2 \tilde b_1)^{2}}
    {(\tilde a_1^2 + \tilde a_2^2)\textbf{I}_{n} + P (\tilde a_1 {\textbf{H}}_2 - \tilde a_2 \textbf{H}_1)^2} \right| \\
    & = \prod_{i=1}^n \frac{(\tilde a_1 \tilde b_2 - \tilde a_2 \tilde b_1)^{2}}{(\tilde a_1^2 + \tilde a_2^2) + P (\tilde a_1 h_{2i} - \tilde a_2 h_{1i})^2}
    \label{eq_optimal_det_Sigma_2}
\end{align}
Plugging the above into (\ref{eq_achievable_rate_b}) gives
\begin{equation}\label{eq_rate2_lim}
    r_l^{(2)} < \lim_{n\rightarrow\infty}\; \frac{1}{2n} \sum_{i=1}^n \log 
    \frac{M(\boldsymbol{\beta},\textbf{h}_i)}{(\tilde a_1 \tilde b_2 - \tilde a_2 \tilde b_1)^{2}} =\frac{1}{2} \mathbb{E}_{\textbf{h}} \left[ \log \frac{\beta_l^{2} { M(\boldsymbol{\beta},\textbf{h}) }}
    {(\tilde a_1 \tilde b_2 - \tilde a_2 \tilde b_1)^{2}} \right].
\end{equation}

For successful decoding, the transmission rate of each user should not exceed the smaller achievable rate of both decoding processes. If $a_l$ or $b_l$ equals zero, meaning one linear combination contains no information of user $l$, the achievable rate of user $l$ simply comes from the other linear combination. Therefore, the proof of Theorem~\ref{thm_achievable_rate} is complete.
\end{proof}

\begin{remark}\label{remark1}
    When we choose $\textbf{a} = (1,0)^T$ and $\textbf{b} = (0,1)^T$, the scheme becomes successive cancellation with lattice codes. In this case, by Theorem~\ref{thm_achievable_rate}, the achievable rates are given by
    \begin{equation*}
        R_1 = \frac{1}{2} \mathbb{E}_{\textbf{h}} \left[\log\frac{1 + P h_{1}^2 + P h_{2}^2}{1 + P h_{2}^2}\right],\quad R_2 = \frac{1}{2} \mathbb{E}_{\textbf{h}}  \left[ \log (1 + P h_{2}^2)\right].
    \end{equation*}
    This is one corner point of the ergodic capacity region. Another corner point can be achieved by choosing $\textbf{a} = (0,1)^T$ and $\textbf{b} = (1,0)^T$. Then with time sharing, the whole ergodic capacity region is achievable. This result matches \cite[Theorem 2]{Hindy15}. 
\end{remark}

\begin{remark}
    When we consider fixed channel gains, where $h_{1i} = h_1,h_{2i}=h_2$ for all $i = 1,\ldots,n$, Theorem~\ref{thm_achievable_rate} will reduce to \cite[Theorem 2]{Zhu17}.
\end{remark}

\section{Capacity achievability}
In this section, we will focus on the two-user Gaussian fast fading MAC with CFMA where the coefficients are chosen as $\textbf{a} = (1,1)$ and $\textbf{b} = (0,1)$ or $(1,0)$. The capacity region can be found in \cite{Gamal2011} and is given by
\begin{equation}
    \label{eq_capacity_ergodic}
    \begin{split}
        R_1 & < \frac{1}{2} \mathbb{E}_{\textbf{h}}  \left[ \log (1 + P h_{1}^2)\right] \\
        R_2 & < \frac{1}{2} \mathbb{E}_{\textbf{h}}  \left[ \log (1 + P h_{2}^2)\right] \\
        R_1 + R_2 & < \frac{1}{2} \mathbb{E}_{\textbf{h}}  \left[ \log (1 + P h_{1}^2 + P h_{2}^2)\right].
    \end{split}
\end{equation}

\subsection{Sufficient and necessary condition}
\begin{theorem}\label{thm_sum_capacity_condition_expectation}
    For the two-user Gaussian fast fading MAC with CFMA, when choosing $\textbf{a} = (1,1)$ and $\textbf{b} = (0,1)$ or $(1,0)$, the sum capacity $C(P,\textbf{h}) = \frac{1}{2}\mathbb{E}_{\textbf{h}}\left[ \log \left(1 + P h_1^2 + P h_2^2\right)\right]$ is achievable if and only if, for some choice of non-zero $\beta_1$, $\beta_2$, it holds that
    \begin{equation}\label{condition_expectation_ergodic}
        \mathbb{E}_{\textbf{h}} \left[\log \frac{f^2(\gamma, \textbf{h})}{\gamma^2(1 + P h_{1}^2 + P h_{2}^2)}\right] \leq 0,
    \end{equation}
    where $\gamma = \beta_1/\beta_2$, and 
    \begin{equation}\label{eq_f_gamma}
        f(\gamma, \textbf{h}) = \gamma^2+1+P(\gamma h_2 - h_1)^2.
    \end{equation}
\end{theorem}
\begin{proof}
    When $\textbf{a} = (1,1)$ and $\textbf{b} = (0,1)$, we will first decode the sum of the inputs and the individual input from user $2$. Following Theorem~\ref{thm_achievable_rate}, the achievable rate of user $1$ is given by
    \begin{align}
        R_1^{(1)} = r_1(\textbf{a},\boldsymbol{\beta}) & = \frac{1}{2} \mathbb{E}_{\textbf{h}} \left[\log 
        \frac{\beta_1^{2} \left(1 + P h_{1}^2 + P h_{2}^2\right)}
        {(\beta_1^2 + \beta_2^2) + P (\beta_1 h_{2} - \beta_2 h_{1})^2}\right]\nonumber\\
        & = \frac{1}{2} \mathbb{E}_{\textbf{h}} \left[\log 
        \frac{\gamma^{2} \left(1 + P h_{1}^2 + P h_{2}^2\right)}
        {f(\gamma, \textbf{h})}\right] \label{eq_R1_b_0}.
    \end{align}
    The achievable rate of user $2$ is given by
    \begin{align}
        R_2^{(1)} & = \min\{r_2(\textbf{a},\boldsymbol{\beta}),\;r_2(\textbf{b}|\textbf{a},\boldsymbol{\beta})\} \nonumber \\
        & = \min\left\{ \frac{1}{2} \mathbb{E}_{\textbf{h}} \left[ \log 
        \frac{\left(1 + P h_{1}^2 + P h_{2}^2\right)}
        {f(\gamma, \textbf{h})}\right],\; 
        \frac{1}{2} \mathbb{E}_{\textbf{h}} \left[\log
        \frac{f(\gamma, \textbf{h})}{\gamma^2}\right] \right\}
        \label{eq_R2_a1_b1}.
    \end{align}
    
    We first prove that if the sum capacity is achievable then (\ref{condition_expectation_ergodic}) holds. If $R_1^{(1)}=0$, the sum-rate becomes $0+R_2^{(1)} \leq r_2(\textbf{a},\boldsymbol{\beta}) < C(P,\textbf{h})$ as $f(\gamma, \textbf{h}) > 1$ for any real $\gamma$ and $\textbf{h}$. So we require $R_1^{(1)}>0$, i.e.,
    \begin{equation}
        \label{eq_R1_positive}
        \mathbb{E}_{\textbf{h}} \left[\log 
        \frac{\gamma^{2} \left(1 + P h_{1}^2 + P h_{2}^2\right)}
        {f(\gamma, \textbf{h})}\right] > 0.
    \end{equation}
    To achieve the sum capacity, given $R_1^{(1)}$ the achievable rate of user $2$ should become
    \begin{equation}
        \label{eq_R2_required}
        R_2^{(1)} = C(P,\textbf{h}) - R_1^{(1)} = 
        \frac{1}{2} \mathbb{E}_{\textbf{h}} \left[\log \frac{f(\gamma, \textbf{h})}{\gamma^2}\right] 
        = r_2(\textbf{b}|\textbf{a},\boldsymbol{\beta}).
    \end{equation}
    Note that $r_2(\textbf{b}|\textbf{a},\boldsymbol{\beta}) > 0$ since
    \begin{equation}
        \label{eq_R2_positive}
        \frac{f(\gamma, \textbf{h})} {\gamma^2} = 1 + \frac{1}{\gamma^2}+P\left(h_2 - \frac{h_1}{\gamma}\right)^2 > 1, \quad \forall \gamma,h_1,h_2\in\mathbb{R}.
    \end{equation}
    Combining (\ref{eq_R2_a1_b1}) and (\ref{eq_R2_required}), we require $r_2(\textbf{b}|\textbf{a},\boldsymbol{\beta}) \leq r_2(\textbf{a},\boldsymbol{\beta})$, i.e.,
    \begin{equation}
        \label{eq_R2_satisfied}
        \mathbb{E}_{\textbf{h}} \left[\log
        \frac{f(\gamma, \textbf{h})}{\gamma^2}\right] \leq \mathbb{E}_{\textbf{h}} \left[ \log 
        \frac{\left(1 + P h_{1}^2 + P h_{2}^2\right)}
        {f(\gamma, \textbf{h})}\right].
    \end{equation}
    The above is equivalent to (\ref{condition_expectation_ergodic}), and with this condition (\ref{eq_R1_positive}) is also satisfied since $f(\gamma, \textbf{h}) > 1$. Thus, if the sum capacity is achievable, we have (\ref{condition_expectation_ergodic}) holds.

    Now we prove that when (\ref{condition_expectation_ergodic}) holds, the sum capacity is achievable. With (\ref{condition_expectation_ergodic}), we have $R_1^{(1)}>0$ and $r_2(\textbf{b}|\textbf{a},\boldsymbol{\beta}) \leq r_2(\textbf{a},\boldsymbol{\beta})$. Therefore, $R_2^{(1)} = r_2(\textbf{b}|\textbf{a},\boldsymbol{\beta})$. The sum-rate will become $r_2(\textbf{b}|\textbf{a},\boldsymbol{\beta}) + r_1(\textbf{a},\boldsymbol{\beta}) = C(P,\textbf{h})$. Thus, the proof for $\textbf{a} = (1,1)$ and $\textbf{b} = (0,1)$ is complete.

    When $\textbf{a} = (1,1)$ and $\textbf{b} = (1,0)$, still following Theorem~\ref{thm_achievable_rate}, we have
    \begin{align}
        R_1^{(2)} & = \min\left\{ \frac{1}{2} \mathbb{E}_{\textbf{h}} \left[ \log 
        \frac{\gamma^{2} \left(1 + P h_{1}^2 + P h_{2}^2\right)}
        {f(\gamma, \textbf{h})}\right],\;
        \frac{1}{2}\mathbb{E}_{\textbf{h}} \left[ \log f(\gamma, \textbf{h})\right] \right\}
        \label{eq_R1_a1_b1}, \\
        R_2^{(2)} & = r_2(\textbf{a},\boldsymbol{\beta}) = \frac{1}{2} \mathbb{E}_{\textbf{h}} \left[ \log \frac{\left(1 + P h_{1}^2 + P h_{2}^2\right)}
        {f(\gamma, \textbf{h})}\right] \label{eq_R2_b_0}.
    \end{align}
    Similar to the previous proof, the achievability of sum capacity is equivalent to $R_1^{(2)} = \frac{1}{2}\mathbb{E}_{\textbf{h}} \left[ \log f(\gamma, \textbf{h})\right] $, thus further equivalent to
    \begin{equation}
        \label{eq_R1_satisfied}
        \frac{1}{2}\mathbb{E}_{\textbf{h}} \left[ \log f(\gamma, \textbf{h})\right]  \leq \frac{1}{2} \mathbb{E}_{\textbf{h}} \left[ \log \frac{\gamma^{2} \left(1 + P h_{1}^2 + P h_{2}^2\right)}
        {f(\gamma, \textbf{h})}\right].
    \end{equation}
    The above is equivalent to (\ref{condition_expectation_ergodic}). Under this condition, $R_2^{(2)}>0$ given (\ref{eq_R2_positive}). Moreover, $R_1^{(2)}>0$ as $f(\gamma, \textbf{h}) > 1$ by definition. Thus, the proof is complete.
\end{proof}

\begin{remark}
    When we consider the fixed channel case and positive $\beta_1,\beta_2$, the condition (\ref{condition_expectation_ergodic}) in Theorem~\ref{thm_sum_capacity_condition_expectation} becomes
    \begin{equation*}
        f(\gamma, \textbf{h}) - \gamma \sqrt{1 + P h_{1}^2 + P h_{2}^2} \leq 0.
    \end{equation*}
    This result matches Case II) in \cite[Theorem 3]{Zhu17}.
\end{remark}

\subsection{A sufficient condition}
Given the statistics of the channel and any fixed $\gamma$, Theorem~\ref{thm_sum_capacity_condition_expectation} already allows us to verify if the sum capacity is achievable. Furthermore, we would like to identify the values of $\gamma$ which allow CFMA to achieve the capacity, as in the fixed channel case. However, it is difficult to work with the expression in Theorem~\ref{thm_sum_capacity_condition_expectation} directly. To this end, we present the following sufficient condition which allows us to answer this question more easily. 
\begin{lem}\label{lem_Sum_Capacity_sufficient}
    For the two-user Gaussian fast fading MAC with CFMA, when choosing $\textbf{a} = (1,1)$ and $\textbf{b} = (0,1)$ or $(1,0)$ with CFMA, the sum capacity is achievable if 
    \begin{equation}
        \label{eq_sum_capacity_sufficient}
        \mathbb{E}_{\textbf{h}} \left[f(\gamma, \textbf{h})\right] \leq |\gamma|\cdot 2^{C(P,\textbf{h})},
    \end{equation}
    where $C(P,\textbf{h}) = \frac{1}{2}\mathbb{E}_{\textbf{h}}\left[ \log \left(1 + P h_1^2 + P h_2^2\right)\right]$.
\end{lem}
\begin{proof}
    Taking logarithm for both sides of ($\ref{eq_sum_capacity_sufficient}$) gives
    \begin{equation}\label{eq_log_sufficient}
        \log(\mathbb{E}_{\textbf{h}} \left[f(\gamma, \textbf{h})\right]) \leq \log|\gamma| + C(P,\textbf{h})
    \end{equation}
    By Jensen's inequality, we have $\log(\mathbb{E}_{\textbf{h}} \left[f(\gamma, \textbf{h})\right]) \geq \mathbb{E}_{\textbf{h}} \left[\log f(\gamma, \textbf{h})\right]$. Combining it with (\ref{eq_log_sufficient}) gives
    \begin{equation}
        \mathbb{E}_{\textbf{h}} \left[\log f(\gamma, \textbf{h})\right] \leq \log|\gamma| + \frac{1}{2}\mathbb{E}_{\textbf{h}}\left[ \log \left(1 + P h_1^2 + P h_2^2\right)\right].
    \end{equation}
    Moving all the terms to the left-hand side and writing them in a single term result in 
    \begin{equation}
        \mathbb{E}_{\textbf{h}} \left[\log \frac{f(\gamma, \textbf{h})}{|\gamma| \sqrt{1 + P h_{1}^2 + P h_{2}^2}}\right] \leq 0.
    \end{equation}
    The above is equivalent to (\ref{condition_expectation_ergodic}). By Theorem~\ref{thm_sum_capacity_condition_expectation}, the sum capacity is achievable.
\end{proof}

We can now give a condition of the channel statistics under which CFMA can achieve the sum capacity with properly chosen $\gamma$. Before that, we define several parameters and functions for present convenience. We let $\rho_l = \sqrt{P}h_l$ and $\mu_l = \mathbb{E}_{\boldsymbol{\rho}} \left[\rho_l\right]$ for $l = 1,2$. Further denote $q(\rho_l) = 1+\mathbb{E}_{\boldsymbol{\rho}}\left[\rho_l^2\right]$ and $C(\boldsymbol{\rho}) = \frac{1}{2}\mathbb{E}_{\boldsymbol{\rho}}\left[ \log \left(1 + \rho_{1}^2 + \rho_{2}^2\right)\right]$.
\begin{theorem}\label{thm_Sum_Capacity_achievability}
    For the two-user Gaussian fast fading MAC with CFMA, when choosing $\textbf{a} = (1,1)$ and $\textbf{b} = (0,1)$ or $(1,0)$, the sum capacity is achievable if
    \begin{enumerate}[{Case} I)]
        \item 
        \begin{equation}
            \label{sumcapacity_achieving_sufficient_condition1}
            g_1(\boldsymbol{\rho}) = \left(\mu_1\mu_2+2^{C(\boldsymbol{\rho})-1}\right)^2 - q(\rho_1) q(\rho_2) \geq 0
        \end{equation}
        when $\gamma$ is chosen within the interval
        \begin{equation}
            \label{eq_gamma_range1}
            \left[ \frac{\mu_1\mu_2+2^{C(\boldsymbol{\rho})-1}-\sqrt{g_1(\boldsymbol{\rho})}}{q(\rho_2)},\quad
            \frac{\mu_1\mu_2+2^{C(\boldsymbol{\rho})-1}+\sqrt{g_1(\boldsymbol{\rho})}}{q(\rho_2)}\right];
        \end{equation}
        \item
        \begin{equation}
            \label{sumcapacity_achieving_sufficient_condition2}
            g_2(\boldsymbol{\rho}) = \left(\mu_1\mu_2-2^{C(\boldsymbol{\rho})-1}\right)^2 - q(\rho_1) q(\rho_2) \geq 0
        \end{equation}
        when $\gamma$ is chosen within the interval
        \begin{equation}
            \label{eq_gamma_range2}
            \left[ \frac{\mu_1\mu_2-2^{C(\boldsymbol{\rho})-1}-\sqrt{g_2(\boldsymbol{\rho})}}{q(\rho_2)},\quad
            \frac{\mu_1\mu_2-2^{C(\boldsymbol{\rho})-1}+\sqrt{g_2(\boldsymbol{\rho})}}{q(\rho_2)}\right].
        \end{equation}
    \end{enumerate}
    
\end{theorem}
\begin{proof}
    We first calculate
    \begin{align}
        \mathbb{E}_{\textbf{h}} \left[f(\gamma, \textbf{h})\right] & = \mathbb{E}_{\boldsymbol{\rho}} \left[ \gamma^2+1+(\gamma \rho_2 - \rho_1)^2\right]\nonumber \\
        & = q(\rho_2)\gamma^2 - 2 \mu_1\mu_2\gamma + q(\rho_1) \label{eq_E_f}.
    \end{align}
    Let $g(\gamma) = \mathbb{E}_{\textbf{h}} \left[f(\gamma, \textbf{h})\right] - |\gamma|\cdot 2^{C(P,\textbf{h})}$, i.e.,
    \begin{equation}
        \label{eq_g_gamma}
        g(\gamma) = q(\rho_2)\gamma^2 - 2 \mu_1\mu_2\gamma - |\gamma|\cdot 2^{C(P,\textbf{h})} + q(\rho_1)
    \end{equation}
    
    In Case I), the condition (\ref{sumcapacity_achieving_sufficient_condition1}) implies $\mu_1\mu_2>0$. With (\ref{eq_gamma_range1}), we observe $\gamma>0$. Then $g(\gamma)$ is given by
    \begin{equation}
        \label{eq_g_gamma1}
        g(\gamma) = q(\rho_2)\gamma^2 - 2 (\mu_1\mu_2+2^{C(P,\textbf{h})-1})\gamma + q(\rho_1).
    \end{equation}
    The condition (\ref{sumcapacity_achieving_sufficient_condition1}) guarantees $g(\gamma)$ has two real roots. Since $q(\rho_2)>0$ and $\gamma$ is chosen between the roots, we have $g(\gamma) \leq 0$. From Lemma~\ref{lem_Sum_Capacity_sufficient}, the sum capacity is achievable in this case.

    In Case II), the condition (\ref{sumcapacity_achieving_sufficient_condition2}) implies $\mu_1\mu_2<0$. With (\ref{eq_gamma_range2}), we observe $\gamma<0$. Then $g(\gamma)$ is given by
    \begin{equation}
        \label{eq_g_gamma2}
        g(\gamma) = q(\rho_2)\gamma^2 - 2 (\mu_1\mu_2-2^{C(P,\textbf{h})-1})\gamma + q(\rho_1).
    \end{equation}
    The condition (\ref{sumcapacity_achieving_sufficient_condition2}) guarantees $g(\gamma)$ has two real roots. Since $q(\rho_2)>0$ and $\gamma$ is chosen between the roots, we have $g(\gamma) \leq 0$. From Lemma~\ref{lem_Sum_Capacity_sufficient}, the sum capacity is achievable in this case. Thus, the proof is complete.
\end{proof}
\begin{remark}
    From the expressions in (\ref{eq_gamma_range1}) and (\ref{eq_gamma_range2}), we see that the optimal choices of $\gamma$ only depend on the channel statistics, but not the channel realizations. This is compatible with our assumption that instantaneous CSI is not available at transmitters.  
\end{remark}
\begin{remark}
    In Theorem~\ref{thm_Sum_Capacity_achievability}, if we replace $\rho_2$ with $-\rho_2$ in Case I), we will get Case II). Thus, we only consider one case in the next analysis while still without loss of generality.
\end{remark}

\subsection{A sub-optimal choice of $\gamma$}
The values of $\gamma$ specified in Theorem~\ref{thm_Sum_Capacity_achievability} are shown to be sufficient for achieving the sum capacity. However, the evaluation of the intervals may not be very easy. The following lemma gives a condition with a simple sub-optimal choice of $\gamma = \mu_1/\mu_2$, i.e., the ratio of the means of the effective channel gains. We will show later in the numerical results that this simple and intuitive choice of $\gamma$ in general only suffers a slight rate loss compared to the optimal choice.
\begin{lem}\label{lem_Sum_Capacity_achievability_simple_gamma}
    For the two-user Gaussian fast fading MAC with CFMA, when choosing $\textbf{a} = (1,1)$, $\textbf{b} = (0,1)$ or $(1,0)$ and $\gamma=\gamma_0 := \frac{\mu_1}{\mu_2}$, the sum capacity is achievable if $\mu_1\mu_2>0$ and
    \begin{equation}
        \label{sumcapacity_achieving_sufficient_condition}
        \frac{\mu_1}{\mu_2} (\text{Var}[\rho_2]+1)  + \frac{\mu_2}{\mu_1} (\text{Var}[\rho_1] + 1) \leq 2^{C(\boldsymbol{\rho})},
    \end{equation}
    where $C(\boldsymbol{\rho}) = \frac{1}{2}\mathbb{E}_{\boldsymbol{\rho}}\left[ \log \left(1 + \rho_{1}^2 + \rho_{2}^2\right)\right]$.
\end{lem}
\begin{proof}
    With the choice of $\gamma = \gamma_0$ in Lemma~\ref{lem_Sum_Capacity_achievability_simple_gamma}, we have
    \begin{align}
        \mathbb{E}_{\textbf{h}} \left[f(\gamma_0, \textbf{h})\right] & = \mathbb{E}_{\boldsymbol{\rho}} \left[ \gamma_0^2+1+(\gamma_0 \rho_2 - \rho_1)^2\right]\nonumber \\
        & = \gamma_0^2+1+ \gamma_0^2 \mathbb{E}_{\boldsymbol{\rho}} \left[ \rho_2^2\right] + \mathbb{E}_{\boldsymbol{\rho}} \left[\rho_1^2\right] - 2 \gamma_0 \mathbb{E}_{\boldsymbol{\rho}} \left[\rho_1\right]\mathbb{E}_{\boldsymbol{\rho}} \left[\rho_2\right] \nonumber\\
        & = \gamma_0^2 (\text{Var}[\rho_2]+1)  + \text{Var}[\rho_1] + 1 \label{eq_e_f}.
    \end{align}
    The condition (\ref{sumcapacity_achieving_sufficient_condition}) gives
    \begin{equation}\label{eq_sufficient_gamma0}
        \gamma_0 (\text{Var}[\rho_2]+1)  + \frac{1}{\gamma_0} (\text{Var}[\rho_1] + 1) \leq 2^{C(\boldsymbol{\rho})}
    \end{equation}
    Combining (\ref{eq_e_f}) and (\ref{eq_sufficient_gamma0}) results in
    \begin{equation}\label{eq_sufficient_gamma0_ef}
    \mathbb{E}_{\textbf{h}} \left[f(\gamma_0, \textbf{h})\right] \leq \gamma_0\cdot 2^{C(P,\textbf{h})}.
    \end{equation}
    Since $\gamma_0 > 0$ by definition, following Lemma \ref{lem_Sum_Capacity_sufficient}, the sum capacity is achievable.
\end{proof}

The above result takes a particularly simple form in the special case when both channel gains are independent and identically distributed (i.i.d.) random variables. By symmetry, $\gamma=1$ is the optimal choice to achieve the sum capacity. The following lemma gives a simplified condition when the channel gains are i.i.d. Gaussian.
\begin{lem}
    If $\rho_1$ and $\rho_2$ are i.i.d. $\sim\mathcal{N} (\mu,\sigma^2)$ with nonzero $\mu$, the sum capacity is achievable if $\gamma = 1$ and 
    \begin{equation}
    \label{sumcapacity_achieving_sufficient_condition_symmetric}
    \sigma^2 \leq 2^{C(\boldsymbol{\rho})-1} - 1.
\end{equation}
\end{lem}
\begin{proof}
    Following Lemma~\ref{lem_Sum_Capacity_achievability_simple_gamma}, with $\mu_1=\mu_2=\mu$ and $\text{Var}[\rho_1] = \text{Var}[\rho_2] = \sigma^2$, we can derive the above lemma.
\end{proof}
\begin{remark}
    We can intuitively see that if the variance is fixed, larger mean value will result in larger $C(\boldsymbol{\rho})$, which helps the condition hold. The numerical results in the next section will explore the values of $\mu$ and $\sigma$ for which (\ref{sumcapacity_achieving_sufficient_condition_symmetric}) holds.
\end{remark}

\subsection{The choices of the decoding coefficients}
We have been focusing on the capacity achievability of CFMA with the coefficient $\textbf{a} = (1,1)$ and $\textbf{b} = (0,1)$ or $(1,0)$. In the fixed channel case, it is shown \cite{Zhu17} that these two sets of coefficients are optimal, in the sense that any other coefficients would require a higher SNR to achieve the capacity region. Hence it is tempting to conjecture that the same result holds for the fading case. However, it is not straightforward because the result in the fixed channel case is proved given the fact that $\gamma$ is directly related to channel gains. In the fading case, as we assume CSIR only, $\gamma$ cannot be changed adaptively with instantaneous channel gains. We have not been able to prove the optimality of these two choices of coefficients in the current paper. However, we present some numerical results which strongly suggest that the same conclusion (namely, the optimal coefficients are $\textbf{a} = (1,1)$ and $\textbf{b} = (0,1)$ or $(1,0)$) holds for the two-user Gaussian fading channels.

Based on Theorem~\ref{thm_achievable_rate}, to achieve the sum capacity, one user should have the achievable rate of (\ref{eq_achievable_rate_MIMO_first}) and the other user should have the rate of (\ref{eq_achievable_rate_MIMO_second}), i.e.,
\begin{equation}\label{eq_general_condition_rate}
    R_1 = r_1(\textbf{a},\boldsymbol{\beta}), R_2 = r_2(\textbf{b}|\textbf{a},\boldsymbol{\beta}) \text{ or } R_1 = r_1(\textbf{b}|\textbf{a},\boldsymbol{\beta}), R_2 = r_2(\textbf{a},\boldsymbol{\beta}).
\end{equation}
In addition, the choice of the decoding coefficients $\textbf{a}$ and $\textbf{b}$ should satisfy
\begin{equation}\label{eq_general_condition_ab}
    (a_1b_2-a_2b_1)^2 = 1. 
\end{equation}
Note that the choice of $\textbf{a}$ will affect both rate expressions (\ref{eq_achievable_rate_MIMO_first}) and (\ref{eq_achievable_rate_MIMO_second}) while $\textbf{b}$ will only affect (\ref{eq_achievable_rate_MIMO_second}). We will now discuss the capacity achieving conditions for different coefficient choices in the following cases. 

\begin{enumerate}[{Case} I)]
    \item 
    When $a_1=0$ or $a_2=0$, from Theory~\ref{thm_achievable_rate} we have $R_1 = r_1(\textbf{b}|\textbf{a},\boldsymbol{\beta})$ or $R_2 = r_2(\textbf{b}|\textbf{a},\boldsymbol{\beta})$. Let us first consider $a_1=0$,  to satisfy (\ref{eq_general_condition_ab}) we require $(a_2b_1)^2=1$, i.e., $a_2,b_1 = \pm 1$. Furthermore, from (\ref{eq_achievable_rate_MIMO_second}) the achievable rate is independent of $b_2$ so $b_2$ could be any integer. An easy and optimal choice is $b_2 = 0$, which directly gives $R_2 = r_2(\textbf{a},\boldsymbol{\beta})$ and achieves sum capacity. By symmetry, when $a_2=0$, we require $a_1,b_2 = \pm 1$. If $b_1=0$, the sum capacity is achievable. Note that the achievable rate pair is independent of $\gamma$ in this case. It can be calculated and is the same as that derived from SIC discussed in Remark~\ref{remark1}, where the corner points of the capacity region can be achieved. 
    \item 
    When there is no zero entry in $\textbf{a}$ and $\textbf{b}$, from Theorem~\ref{thm_achievable_rate} we have 
    \begin{equation}
        \label{eq_achievable_rate_nonzero_ab}
        R_l = \min\{r_l(\textbf{a},\boldsymbol{\beta}),\;r_l(\textbf{b}|\textbf{a},\boldsymbol{\beta})\},\text{ for } l = 1,2.
    \end{equation}
    To achieve the sum capacity, the choice of $\textbf{a}$ and $\textbf{b}$ should satisfy (\ref{eq_general_condition_ab}). Moreover, from (\ref{eq_achievable_rate_MIMO_first}), (\ref{eq_achievable_rate_MIMO_second}), (\ref{eq_general_condition_rate}) and (\ref{eq_achievable_rate_nonzero_ab}) we have
    \begin{equation}
        \label{eq_condition_ab_nonzero}
        \mathbb{E}_{\textbf{h}} \left[\log \frac{\left\{a_1^2\gamma^2+a_2^2+P(a_1\gamma h_2-a_2h_1)^2\right\}^2}{\gamma^2(1 + P h_{1}^2 + P h_{2}^2)}\right] = 0.
    \end{equation}
    This is a very strong equality condition and is hard to be satisfied. Even if this condition can be satisfied, it is expected that only a few isolated choices of $\gamma$ can satisfy it, which means only a few points on the boundary of the capacity region can be achievable. In this sense, the coefficient choice in this case is usually not applicable.
    \item 
    When $a_1,a_2 \neq 0$ and $b_1 = 0$ or $b_2 = 0$, to achieve the sum capacity we first require (\ref{eq_general_condition_ab}) to be satisfied, i.e., $a_1,b_2 = \pm 1$ when $b_1 = 0$ and $a_2,b_1 = \pm 1$ when $b_2 = 0$. Since $\textbf{b}$ will not affect the achievable rates as long as (\ref{eq_general_condition_ab}) is satisfied, we can choose $\textbf{b} = (0,1)$ or $\textbf{b} = (1,0)$ without loss of generality. By further observation of (\ref{eq_Mh}), we can find that the sign of either $a_1$ or $a_2$ will not change the achievable rate if we adaptively choose the sign of the other coefficient of $\textbf{a}$. Thus we choose $a_1 = 1$ when $\textbf{b} = (0,1)$ and $a_2 = 1$ when $\textbf{b} = (1,0)$ without loss of generality. From Theorem~\ref{thm_achievable_rate} and (\ref{eq_general_condition_rate}), we also require
    \begin{equation}
        \label{eq_condition_b_zero}
        \mathbb{E}_{\textbf{h}} \left[\log \frac{\left\{a_1^2\gamma^2+a_2^2+P(a_1\gamma h_2-a_2h_1)^2\right\}^2}{\gamma^2(1 + P h_{1}^2 + P h_{2}^2)}\right] \leq 0.
    \end{equation}
    To summarize, to achieve sum capacity in this case, the decoding coefficients should be $\textbf{a} = (1, a_2), \textbf{b} = (0,1)$ and $\textbf{a} = (a_1, 1), \textbf{b} = (1,0)$ for some integer $a_1,a_2$ while there exists some $\gamma$ that satisfies (\ref{eq_condition_b_zero}). It is expected that a range of $\gamma$ would satisfy the condition and a segment of the capacity boundary could be achieved. Since we suggest $\textbf{a} = (1,1)$ is the best coefficient, the question is whether any choices of $a_1\neq 1$ or $a_2\neq 1$ could benefit the satisfaction of (\ref{eq_condition_b_zero}) or enlarge the achievable capacity region.
\end{enumerate} 

To see how the choice of $\textbf{a}$ can affect the achievable rate region in Case III), we study an example where $h_1\sim \mathcal{N}(10,4)$ and $h_2\sim \mathcal{N}(20,9)$ are independent Gaussian. The decoding coefficient $\textbf{b}$ is either $(1,0)$ or $(0,1)$. The left plot of Fig.~\ref{fig_ergodic_a_eg} shows the channel capacity and achievable rate region for $\textbf{a} = (1,1), (1,2)$ and $(2,1)$, respectively. It can be observed that the choice of $\textbf{a} = (1,1)$ achieves a much larger capacity region than other choices. In the right plot of Fig.~\ref{fig_ergodic_a_eg}, the left-hand side of (\ref{eq_condition_b_zero}) is plotted as a function of $\gamma$ for different $\textbf{a}$ choices. The function with $\textbf{a} = (1,1)$ has a much larger range of $\gamma$ that can achieve the sum capacity than other $\textbf{a}$ choices. The same observation holds for many other channel statistics according to our simulation results, which are not presented here. All these simulations strongly suggest that $\textbf{a}=(1,1)$ is the best decoding coefficient choice. This is also the reason why we focus on the case where $\textbf{a} = (1,1), \textbf{b}=(0,1)$ or $(1,0)$ in the previous sections.
\begin{figure}[!ht]
    \centering
    \includegraphics[width = 0.8\linewidth]{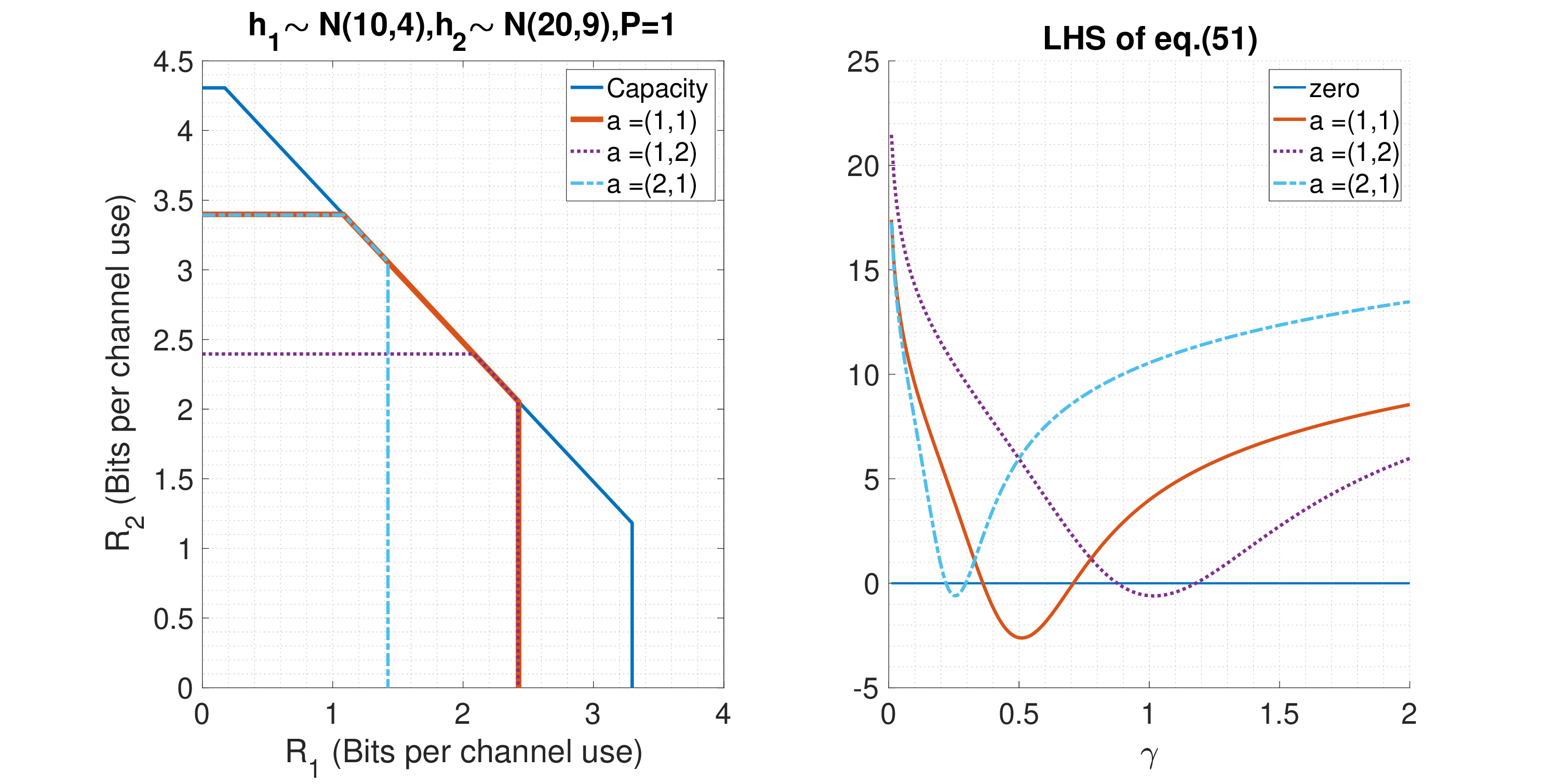}
    \caption{\small The achievable rate region and the capacity-achieving conditions for $\textbf{a} = (1,1), (1,2), (2,1)$. The choice of $\textbf{a} = (1,1)$ achieves the largest region of the capacity and has the largest range of $\gamma$ that can achieve the sum capacity.}
    \label{fig_ergodic_a_eg}
\end{figure}

\section{Numerical results}
We first look at a simple case where both channel gains are i.i.d. Gaussian with mean $2$ and variance $0.25$. The left plot of Fig.~\ref{fig_ergodic_eg_1} shows the capacity region in blue and the achievable rates with $\textbf{a} = (1,1),\textbf{b} = (0,1)$ and $\textbf{a} = (1,1),\textbf{b} = (1,0)$ in red and yellow, respectively. The achievable rate region is derived according to Theorem~\ref{thm_achievable_rate} by varying $\beta_1,\beta_2$. We can observe that the majority part of the capacity region, including a large part of the dominant face, is achievable. The achievable rates in red and yellow are almost the same. But the same rate pairs are achieved by different pair of $\beta_1,\beta_2$. The right plot shows two functions of $\gamma = \beta_1/\beta_2$. The blue curve refers to $\mathbb{E}_{\textbf{h}} \left[\log \frac{f^2(\gamma, \textbf{h})}{\gamma^2(1 + P h_{1}^2 + P h_{2}^2)}\right]$ from Theorem~\ref{thm_sum_capacity_condition_expectation} that gives the sufficient and necessary condition for capacity-achieving $\gamma$, while the red curve represents $\mathbb{E}_{\textbf{h}} \left[f(\gamma, \textbf{h})\right] - |\gamma|\cdot 2^{C(P,\textbf{h})}$ from Lemma~\ref{lem_Sum_Capacity_sufficient} that gives a sufficient condition for capacity-achieving $\gamma$. We are interested in the non-positive parts of both curves. The values of $\gamma$ which produce non-positive function values in the blue curve are sum capacity achieving choices. It can be observed that the range of $\gamma$ which produces non-positive function value in the red curve is within that of the blue curve, which verifies that Lemma~\ref{lem_Sum_Capacity_sufficient} presents a sufficient condition of Theorem~\ref{thm_sum_capacity_condition_expectation}.
\begin{figure}[!ht]
    \centering
    \includegraphics[width = 0.8\linewidth]{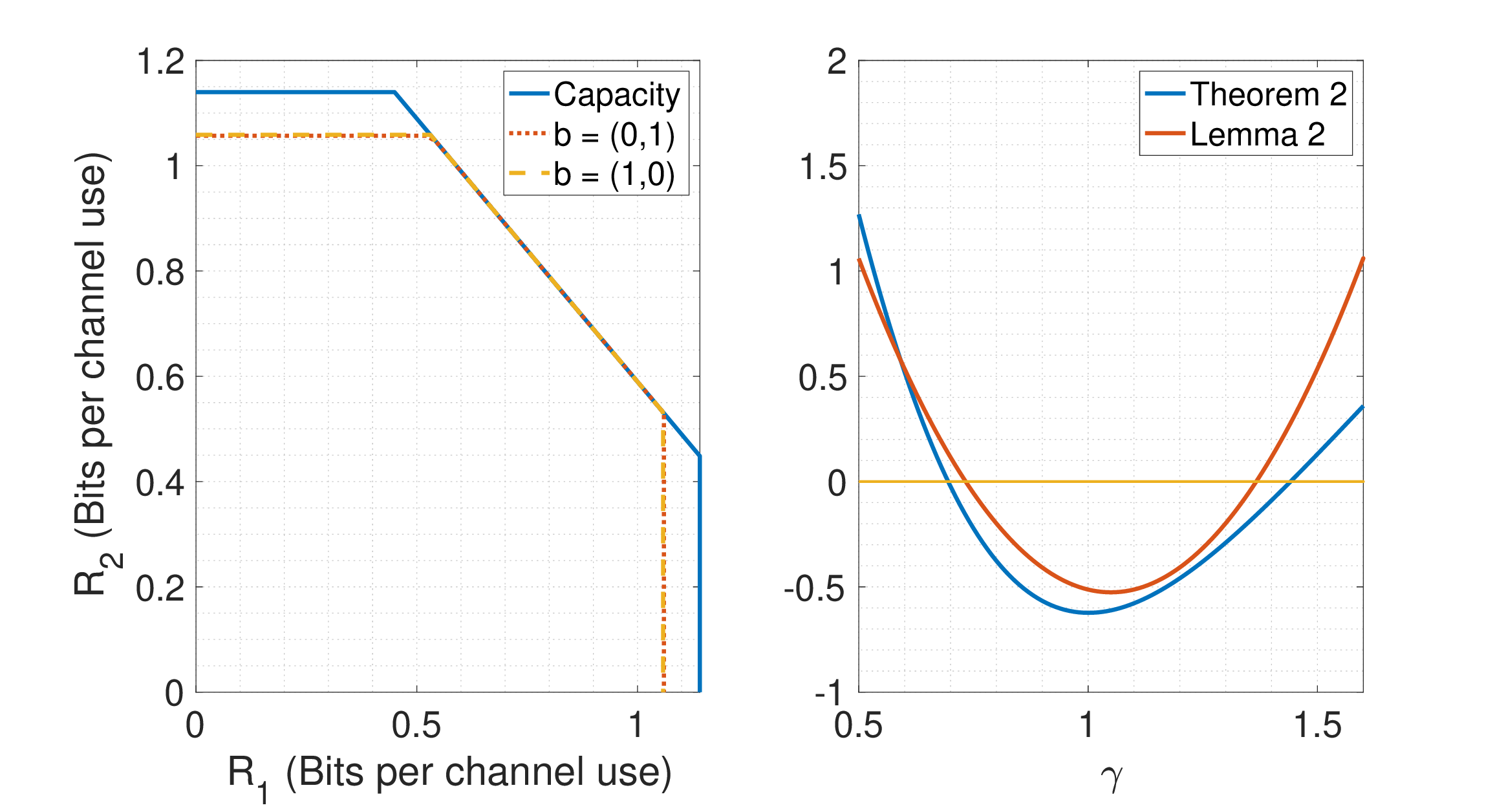}
    \caption{\small Symmetric Gaussian channel gains with $\mu=2,\sigma=0.5$. The left plot shows the majority part of the capacity region is achievable with $\textbf{a} = (1,1)$ and $\textbf{b} = (1,0)$ or $(0,1)$; The right plot shows the functions of $\gamma$ in Theorem~\ref{thm_sum_capacity_condition_expectation} and Lemma~\ref{lem_Sum_Capacity_sufficient}. We are interested in the $\gamma$'s which produce non-positive function values because with those choices of $\gamma$ the sum capacity is achievable.}
    \label{fig_ergodic_eg_1}
\end{figure}

We then look at how the variance of the channel gains affects the capacity achievability. In Fig.~\ref{fig_ergodic_rate_region}, we fix the mean of the channel gains $\mu$ and vary the standard deviation $\sigma$ from $0$ to $0.85$. It can be observed that when $\sigma = 0$ which is the fixed channel case, the whole capacity region is achieved. However, the entire capacity region is not achievable with the equation coefficients $\textbf{a} = (1,1),\textbf{b} = (0,1)$ or $(1,0)$ as long as the variance becomes positive. With $\sigma$ increasing, it becomes harder to achieve the whole capacity region. In particular, when $\sigma = 0.85$ even the sum capacity is not achieved. 
\begin{figure}[!ht]
    \centering
    \includegraphics[width = 0.9\linewidth]{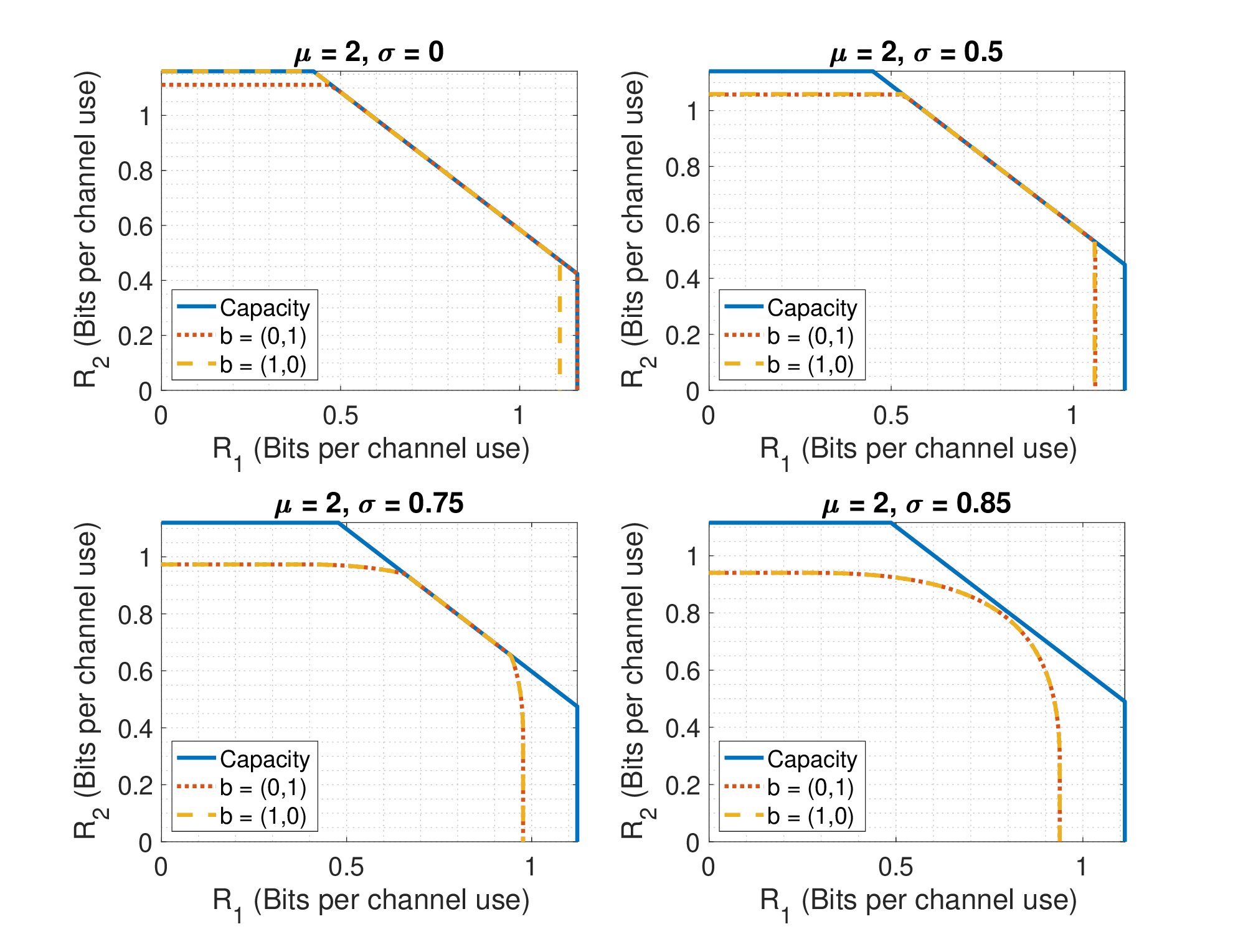}
    \caption{\small Symmetric Gaussian channel gains with $\mu=2$ and $\sigma=0,0.5,0.75.0.85$. The first subplot represents the fixed channel case. With the standard deviation increasing, the capacity becomes harder to achieve. In the last subplot, the achievable rate region falls off the boundary of the capacity boundary.}
    \label{fig_ergodic_rate_region}
\end{figure}


Next we investigate how the mean and variance of the channel gains affect the sum capacity achievability in the Gaussian case. In the case where $\rho_1$ and $\rho_2$ are i.i.d. Gaussian, the left plot of Fig.~\ref{fig_ergodic_symmetric_gaussian} shows how $\mu$ and $\sigma$ affect the sum capacity achievability. Region I refers to the cases where the sum capacity is not achievable with any choice of $\gamma$. Region II represents the cases where the sum capacity is achievable when $\gamma$ is chosen optimally to satisfy Theorem~\ref{thm_sum_capacity_condition_expectation}. Recall that since Theorem~\ref{thm_sum_capacity_condition_expectation} gives the necessary and sufficient conditions for capacity-achieving $\gamma$, Region II is the complement of Region I. It can be observed that larger $\mu$ and smaller $\sigma$ will benefit the sum capacity achievability. In other words, if the channel gains are too uncertain, it will be difficult to achieve capacity, which makes sense. The right plot of Fig.~\ref{fig_ergodic_symmetric_gaussian} gives a glimpse of the asymptotic relationship between the mean ($\mu$) and variance ($\sigma^2$) of the channel gain. The numerical result seems to suggest that the sum capacity can be achieved if $\sigma^2 < 2\mu$.
\begin{figure}[!ht]
    \centering
    \includegraphics[width = 0.9\linewidth]{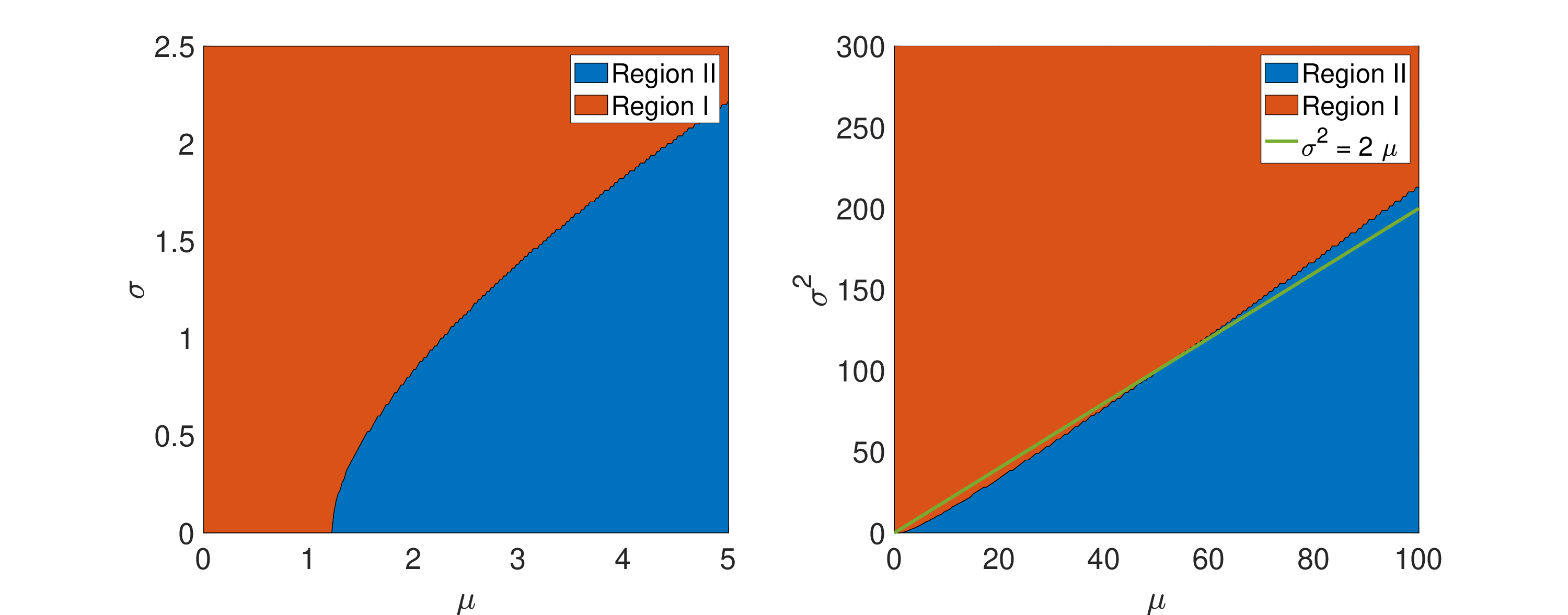}
    \caption{\small Symmetric Gaussian channel gains. In this case, the channel gains of both users are i.i.d. With the mean and standard deviation shown in the right bottom region, the sum capacity is achievable with $\textbf{a} = (1,1)$ and $\textbf{b} = (1,0)$ or $(0,1)$. This indicates that it is good to have large mean and small variance. Asymptotically, to achieve the sum capacity, the mean is required to increase faster than the standard deviation. The right plot suggests that the sum capacity can be achieved if $\sigma^2 < 2\mu$.}
    \label{fig_ergodic_symmetric_gaussian}
\end{figure}

When $\rho_1$ and $\rho_2$ are Gaussian with different distributions, we first fix their variances and look at how their means affect the sum capacity achievability. In Fig.~\ref{fig_ergodic_asymmetric_gaussian_fixed_var}, we set $\sigma_1^2 = \sigma_2^2 = 0.25$. Region I refers to the cases where the sum capacity is not achievable with any choice of $\gamma$. Region III contains the cases where the sum capacity is achievable with the specific choice $\gamma = \mu_1/\mu_2$. Region IV + III represents the cases where the sum capacity is achievable when $\gamma$ is chosen optimally, which is the complement of Region I. We can find that $\gamma = \mu_1/\mu_2$ is not a bad suboptimal choice of $\gamma$ since Region IV is not big. When $\mu_1=\mu_2$, this choice is optimal since the gap between Region I and Region III vanishes. It can be observed that if the mean of one channel gain is very small compared to its variance, even the mean of the other channel gain is large enough, the sum capacity is not achievable with the equation coefficients $\textbf{a} = (1,1)$ and $\textbf{b} = (0,1)$ or $(1,0)$ for any choice of $\gamma$. In this case, we may require other choice of equation coefficients. For example, the classical successive interference cancellation (corresponding the coefficients $\textbf{a} = (1,0)$ and $\textbf{b} = (0,1)$) still achieve the sum capacity, though only the corner point on the dominant face. 
\begin{figure}[!ht]
    \centering
    \includegraphics[width = 0.5\linewidth]{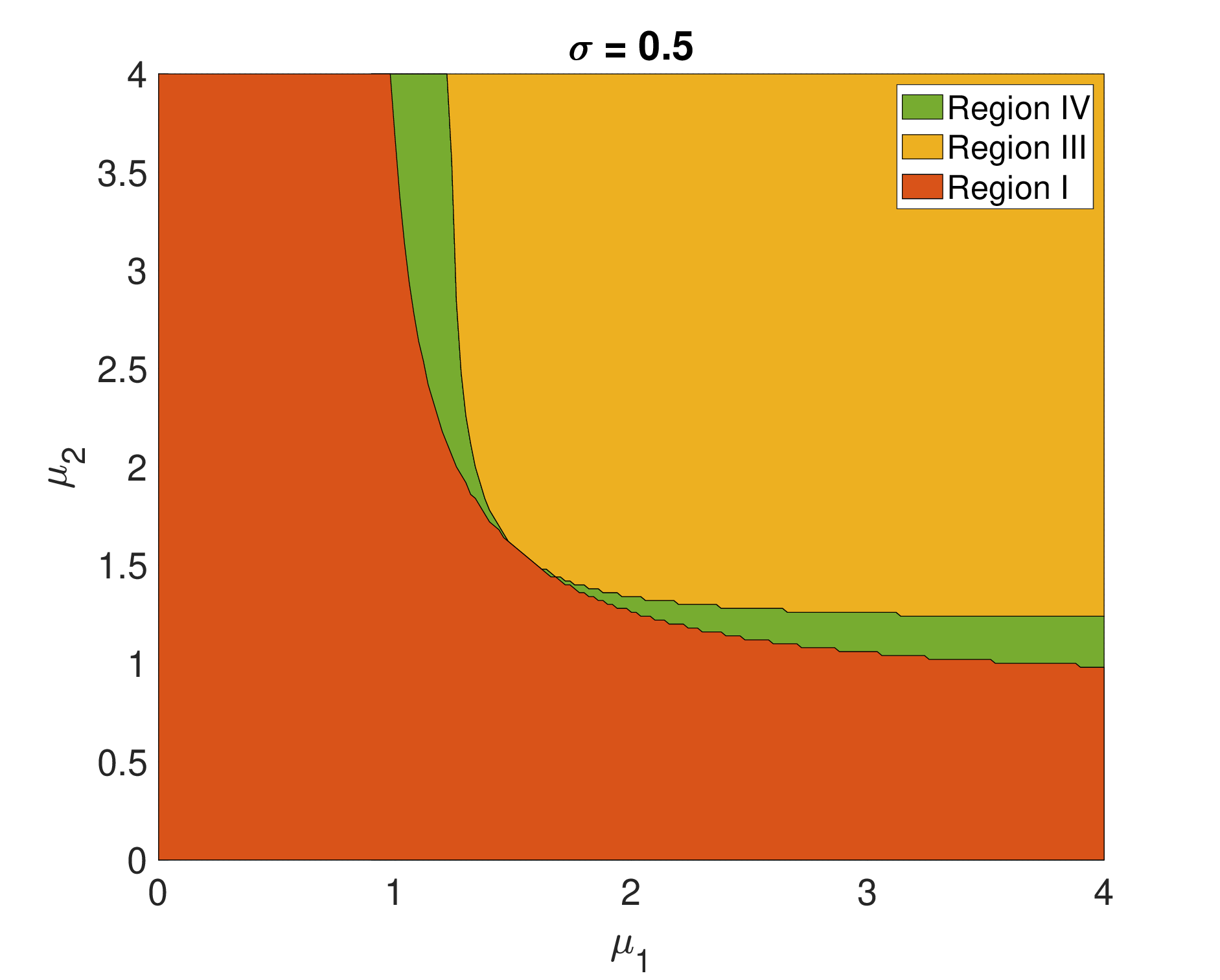}
    \caption{\small Asymmetric Gaussian channel gains with standard deviation equaling $0.25$. The standard deviations of both channel gains are fixed. In the most right top region (Region III), the sum capacity is achievable with $\gamma = \mu_1/\mu_2$. If we can choose the best $\gamma$, this region will become larger as Region IV + Region III. But the difference between them is small.}
    \label{fig_ergodic_asymmetric_gaussian_fixed_var}
\end{figure}

Still in the case where $\rho_1$ and $\rho_2$ are Gaussian with different distributions, we fix the means of both channel gains and investigate how their standard deviations affect the sum capacity achievability. Fig.~\ref{fig_ergodic_asymmetric_gaussian_fixed_mean} shows how their standard deviations affect the sum capacity achievability when $\mu_1=\mu_2 = 2$. Region I refers to the cases where the sum capacity is not achievable with any choice of $\gamma$. Region III stands for the cases where the sum capacity is achievable when $\gamma = \mu_1/\mu_2$. Region IV + III represents the cases where the sum capacity is achievable when $\gamma$ is chosen optimally. Since Region IV is quite small, $\gamma = \mu_1/\mu_2$ is a good suboptimal choice in this case. It becomes optimal when $\sigma_1=\sigma_2$ since the gap vanishes.
\begin{figure}[!ht]
    \centering
    \includegraphics[width = 0.5\linewidth]{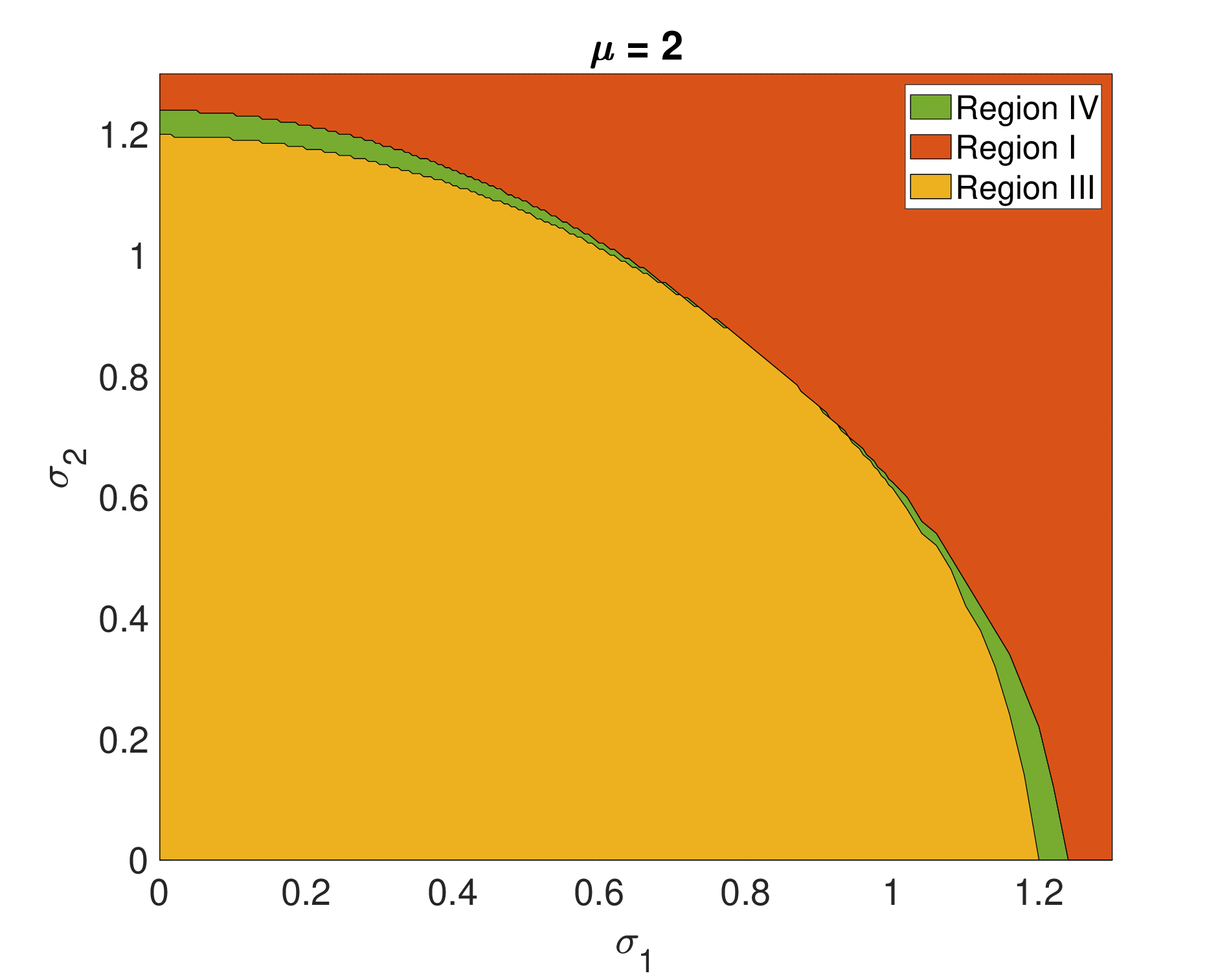}
    \caption{\small Asymmetric Gaussian channel gains with mean equals to $2$. The means of both channel gains are fixed. In the most left bottom region (Region III), the sum capacity is achievable with $\gamma = \mu_1/\mu_2$. If we can choose the best $\gamma$, this region will become larger as Region IV + Region III. But the gap is small.}
    \label{fig_ergodic_asymmetric_gaussian_fixed_mean}
\end{figure}

\section{Conclusion}
In this paper, we study the capacity achievability of CFMA in the two-user fading MAC. It is found that the channel statistics heavily affect the capacity achievability. In contrast to the fixed channel case where the entire capacity region is achievable with large enough transmission power, it is not the case in the fading scenario if the channel variance is large compared to its mean. We focus on the equation coefficients $\textbf{a} = (1,1)$ and $\textbf{b} = (0,1)$ or $(1,0)$, which is strongly suggested to be the best decoding coefficient choice by the numerical examples. We see for many channel statistics, CFMA can achieve a large range of points on the capacity boundary However, we can still see that the sum capacity cannot be achieved for some channel statistics. In that case, we can apply SIC ($\textbf{a} = (1,0),\textbf{b} = (0,1)$ and $\textbf{a} = (0,1),\textbf{b} = (1,0)$) to achieve both corner points of the capacity pentagon to achieve the sum capacity.


\bibliographystyle{IEEEtran}
\bibliography{reference,IEEEabrv}

\end{document}